\documentclass[11pt]{article}

\usepackage{mathtools, nicefrac, wasysym, makecell}
\usepackage[margin=1in]{geometry}
\usepackage{amsmath, amssymb, amsthm}
\usepackage{graphicx}
\usepackage[natbib=true, style=authoryear, backend=biber]{biblatex}
\usepackage{hyperref}
\newtheorem{theorem}{Theorem}
\newtheorem{lemma}[theorem]{Lemma}

\theoremstyle{remark}
\newtheorem{remark}{Remark}
\newcommand{\keywords}[1]{\par\medskip\noindent\textbf{Keywords:} #1}

\newtheorem{proposition}[theorem]{Proposition}%

\newcommand{\invCb}{C^*_{B}}
\newcommand{\invCa}{C^*_{A}}

\usepackage{booktabs}
\usepackage{xcolor}

\newif\ifshowchanges
\showchangesfalse   

\hypersetup{hidelinks}

\title{Addressing the Between-Group Comparison Problem: Detecting Differences Between Correlation Matrix Populations due to Single-variable Perturbations for Resting State fMRI}

\author{Itamar Faran$^{1},$ Michael Peer$^{2},$ Shahar Arzy$^{3},$ 
Yuval Benjamini$^{1\ast}$\\
{\small $^{1}$Department of Statistics and Data Science, Hebrew University of Jerusalem, Israel}\\
{\small $^{2}$Department of Psychology, University of Pennsylvania, PA, USA}\\
{\small $^{3}$Department of Medical Neurobiology, Hebrew University of Jerusalem, Israel}\\
{\small $^\ast$Correspondence: yuval.benjamini@mail.huji.ac.il}
}
\date{}

\begin{document} 

\maketitle

\begin{abstract}
Resting-state fMRI (RS-fMRI) has been known for decades as a promising method for evaluating cognitive and mental states, both in health and especially in disease, due to its ease of implementation as a short, standard MRI protocol. In clinical settings, a group of patients with a given disorder is typically compared to a group of healthy controls. This poses an inherent challenge of between-group comparison. We propose a new efficient model for characterizing changes to the temporal synchronization of brain activity measured using RS-fMRI between groups, summarized as individual correlation matrices. Our model posits that the between-group differences are the product of single-region effects describing the increase or decay of synchronization with the rest of the brain. This parsimonious model pools the correlation coefficients of each region with all others, and therefore can detect differences between groups even in small samples. Inference for this model accounts for the variability in individual correlation matrices, the within-group differences across individuals, and for the approximation error of the single-region model. This results in per-region estimates and confidence intervals for the parameters governing the difference between groups.  In simulations, our model shows increased power to detect model-aligned alternatives compared with competing approaches. To demonstrate feasibility of the method in a clinical application, we use the model to analyze RS-fMRI correlation matrices in patients with transient global amnesia and healthy controls. Our model detects significant decreases in synchronization for the patient population in the amygdala after multiplicity correction as well as borderline decreases in memory-related brain regions that were not detected using mass-univariate tests without prior knowledge, suggesting its usefulness in the application of RS-fMRI in clinical settings.
\end{abstract}

\newpage
\section{Introduction}
Differences between healthy and patient populations in the temporal synchronization of neural activity across brain regions are often used for characterizing neuropsychiatric disorders \citep{biswal2025history}. These differences are measured using resting-state functional MRI (RS-fMRI), which can identify functional information even when the disturbance is not apparent structurally. Importantly, RS-fMRI is a short protocol that can be added to a standard MRI examination and does not require any patient cooperation. It therefore emerged as a potential source of both neuro-cognitive discoveries and neuro-psychiatric clinical implications. 

To study such differences, the temporal sequences measured by RS-fMRI are usually aggregated into correlation coefficients of the BOLD (Blood-oxygen-level-dependent) signal between different brain regions, resulting in person-specific correlation matrices that represent each individual's correlation (or \emph{functional connectivity}) across brain regions \citep{lee2013resting}. In some cases, the differences can be characterized as a reduced association between some brain regions and the rest of the brain \citep{peer2014reversible}. These differences can occur due to a decay in the synchronized activity or an increase in non-synchronized activity of that region. We term these events \emph{single variable perturbation}s.

Despite their widespread use, comparative functional neuroimaging studies (both RS- and task-fMRI) that compare groups of interest to healthy controls are notoriously challenging. In recent years, methods for RS-fMRI have focused on increased flexibility to better use longer data collection via time-evolution models \citep{vidaurre2018discovering,liu2018co,lindquist2014evaluating} and larger datasets ($n>1000$) that can support models based on deep neural networks \citep{malkiel2022self,kim2021representation}. However, for the relatively small cohorts common in clinical settings, it remains difficult to correctly identify and estimate differences, even when restricting the analysis to correlations between prespecified regions. Standard approaches to between-group statistical testing in functional neuroimaging encounter substantial difficulties in controlling false positive rates while maintaining statistical power. A key difficulty arises from the high dimensionality of neuroimaging data: whole-brain analyses typically involve hundreds of thousands of voxels or vertices, necessitating stringent multiple comparison corrections that substantially reduce sensitivity to true group differences. Non-parametric and permutation-based approaches offer potential solutions but require careful implementation, particularly in handling confounding variables, spatial autocorrelation, and the construction of appropriate null distributions. Furthermore, permutation-based procedures are often less powerful in detecting differences, especially in small samples.

In statistical research, methods for detecting differences between correlation matrices in general have been studied extensively for the two-sample design \citep{lawley1963testing,aitkin1968tests, jennrich1970asymptotic},
with recent methods tackling the case where the sample is not large compared to the matrix dimension.
Most methods focus either on detecting global differences in the full matrix \citep{schott2007testing, zheng2019test},
or on detecting large \emph{local} differences in individual correlation-coefficients \citep{cai2013two}. 
Additional strategies for dealing with the high-dimensionality include ridge regularization of the difference matrix or adaptive soft thresholding \citep{cai2016inference, cai2017global, na2019estimating}.
However, whereas global tests cannot characterize the nature of the difference, 
local tests that focus on single-coefficients suffer from low power when the signal is weak and diffuse. 
Specialized methods can identify weaker signals by focusing on intermediate alternatives such as differences in cliques of variables \citep{bodwin2018testing}, on the leading eigenvector of the differential matrix \citep{ZhuTesting}, or in sub-networks \citep{guo2008unified}. 

An additional challenge of the two-group design is that each condition is represented by several correlation matrices representing different units, with non-negligible differences between units in the group. Interesting between-group effects need to exceed the expected effects of this within-group variability, but most statistical methods discussed above ignore it. 
Common neuroimaging methods therefore turn to standard mass-univariate approaches,
calculating a two-group $t$-statistic for each correlation coefficient \citep{baker2014disruption, brier2012loss, liu2014impaired}.
The resulting p-values or intervals are then corrected to account for multiple comparisons.
Because the number of tests grows quadratically with the dimension of the correlation matrix,
the corrections for multiplicity reduce the statistical power.
This loss of power is considerable when working with small samples,
where power is low to begin with. To address this, researchers often reduce the number of regions considered using coarse-graining parcellations \citep{MEZER20091117,peer2014reversible,saadon2015discontinuity,glasser2016multi,monsa2018self,van2019cerebral,dhamala2023one}.

In this paper, we propose a model attuned to detecting single-variable perturbations that affect the correlation between the groups, while accounting for the variation within the group. In the context of correlation matrices,
single-variable perturbations manifest as changes in the full rows (and full columns) corresponding to the perturbed variables.
These differences would typically be directional: if noise is added to a variable in one condition, or the variable's signal is disturbed, that variable's correlations with all other variables would be smaller in the perturbed condition. 

Motivated by this principle, we propose a parsimonious model for expressing differences 
between correlation matrices. The model posits that the case group (B) mean correlation coefficient between two regions $\rho^B_{jk}$ can be 
expressed in terms of (a) the control group mean correlation $(\rho_{jk}^A)$, and (b)  directional \emph{perturbation} effects, one effect parameter per variable $(\alpha_j>0)$. A multiplicative link function describes how the control group correlation and the effects of the two variables interact to produce the case group correlation $\rho_{jk}^B$:
\begin{equation}
\rho_{jk}^B = \rho_{jk}^A \cdot \alpha_j \cdot \alpha_k . 
\label{eq:MultLink}
\end{equation}
That is, the divergence of the correlation matrices in index $j,k$ is determined
by the interaction of the correlation divergence (decay or increase) of rows and columns $j$ and $k$.
Finding $\alpha_{j} < 1$ means that the correlations of the $j$'th variable with other variables have decayed (shrunk towards zero) in the case group, compared to the control group. This encodes the assumption that the decay (or increase) in correlations is proportional to the original magnitude of the correlation \-- originally stronger absolute correlations of the perturbed variables will suffer from greater absolute decay. That is, $\alpha_j<1$ represents a decay in correlations for the $j$'th variable, but the magnitude of this decay depends on the correlations in the control group. We find that the vector of decay parameters ($\alpha$) in the multiplicative model can efficiently describe local differences between the groups when the expected correlation matrices are dense. 
Using this parameterization, the dimension of $\alpha$ (and the number of comparisons) grows \emph{linearly} with the dimension of the correlation matrix, making the model suitable for small sample sizes.

Because the model is only an approximation to the differences between the correlations, our inference strategy accounts for the effect of modeling error and within group variation on the estimators, by using the residuals of each subject from the model group average. The variables' effects are estimated by finding the roots of a system of estimation equations that accounts for the asymptotic covariance of correlation matrices \citep{nel1985matrix}. 
We develop robust estimators of the standard deviations using a generalized estimating equations (GEE) framework \citep{liang1986longitudinal},
or a less parametric but more computationally intensive two-sample Jackknife procedure \citep{efron1981jackknife}.
Both of these approaches are designed to measure population variability.
Confidence intervals and p-values are based on the approximate Gaussian distribution of GEE solutions. 

We demonstrate the efficiency of our model on resting-state functional MRI (RS-fMRI) data gathered from patients with a unique and rare phenomenon, transient global amnesia (TGA). Patients diagnosed with TGA present mainly an acute and complete loss of ability to retain new information (\emph{anterograde amnesia}) and a gradually declining loss of memories to the recent time (\emph{retrograde amnesia}) that resolves autonomously within $\sim24$  hours. The acute phase of the disturbance resolves rapidly, therefore minimizes the chances to capture it with MRI. 
In recent years, it was discovered that small well-defined lesions may be evident in TGA patients in a certain MRI protocol in a sub-region of the hippocampus,
a brain structure that plays a major role in memory processes
\citep{sedlaczek2004detection, bartsch2006selective, bartsch2007evolution}.
Furthermore, a study of RS-fMRI recorded within the acute stage revealed a reduction of correlations in an extensive memory network including the hippocampi \citep{peer2014reversible}.
These disturbances were found at the single-patient level,
and their magnitude correlated with the clinical severity of the impairment, findings that were not possible using merely lesion detection methods.
The restricted structural disturbance of TGA and the small sample size (due to the rapid transient character of TGA and its rarity) serve as a unique setting to test our model.

\section{Methods}

\subsection{Notations \& Preliminaries}
\label{section:preliminaries}
Let $X_i$ be the observed data matrix for subject $i$ of size $T_i \times p$ measuring $p$ variables (columns) at $T_i$ points (rows).
In the fMRI data, the $p$ columns of $X_i$ correspond to $p$ brain regions, sampled at $T_i$ time points (the rows).
We assume that the measurements (rows) of $X_i$ are centered (mean $\boldsymbol{0}$), and that $T_i$ is constant over all subjects ($T_i=T$). 
For each subject, the measurement matrix $X_i$ is summarized into the empirical correlation matrix $R_i$.
Define the empirical covariance operator $cov\left(X_i\right)\coloneqq T^{-1}X_i^tX_i$,
and the scaling operator on $W_i=cov\left(X_i\right)$ to be
$scale\left(W_i\right) \coloneqq diag\left(W_i\right)^{-\nicefrac{1}{2}}\cdot W_i\cdot diag\left(W_i\right)^{-\nicefrac{1}{2}}$. (We use lower-case for matrix operators).
The empirical correlation operator on $X_i$ is $corr\left(X_i\right) \coloneqq scale\left(cov\left(X_i\right)\right)$.
Then $R_i = corr(X_i) \in \mathcal{C}_{p}$ is the observed correlation matrix for subject $i$, 
where $\mathcal{C}_{p}$ is the set of symmetric positive semi-definite $p\times p$ matrices with a unity diagonal.
Since correlation matrices are symmetric, they can be fully characterized using the $m=\binom{p}{2}=p\left(p-1\right)\slash2$ elements of the upper (or lower) triangle,
so we vectorize them in some of the formulas.

Define $\text{vec}:\mathbb{M}_{p}\left(\mathbb{R}\right)\mapsto\mathbb{R}^{m}$ as the vectorization operator with
$\text{vec}\left(R\right) \coloneqq \left(R_{12},R_{13,}\cdots,R_{1p},R_{23},\cdots R_{p-1,p}\right)^{t}$,
or alternatively $\vec{R} \coloneqq \text{vec}\left(R\right)$. 
We denote the Hadamard (element-wise) product between matrices by $\circ$, and the first and second moments of random vectors or matrices by $E[\cdot]$ and upper-case $Cov(\cdot)$. We denote the identity matrix of dimension $m$ by $Id_m$.  

\subsection{A Model for Differences Between Correlation Matrices}
\label{section:model}

We study a two-group design, where $n_A$ units (or subjects) are observed from the control group ($\mathcal{A}$) and $n_B$ are observed from the case group ($\mathcal{B}$). For each unit, a $p$-dimensional measurement vector is sampled repeatedly ($T$ times), and these are summarized as a correlation matrix $R_i\in \mathcal{C}_p$. The goal of the researcher is to identify and characterize differences between the expected case correlation matrix $\Lambda^{B}$ and the expected control correlation matrix $\Lambda^{A}$.

\subsubsection{Model overview}

Our model describes the difference between groups using a single p-dimensional effect vector:
\begin{equation}
\Lambda^A :=  E_{i\in \mathcal{A}}[R_i] =\Theta, \qquad \Lambda^B :=  E_{i\in \mathcal{B}}[R_i] = g(\Theta, \alpha),
\label{eq:meanModel}
\end{equation}
for correlation matrix  $\Theta\in \mathcal{C}_p$, $p$-dimensional parameter vector $\alpha \in \Omega_\alpha \subset R_+^p$, and a link function $g:\mathcal{C}_p\times \Omega_\alpha \to \mathcal{C}_p$. 
We can therefore write: 
\begin{equation}
    R_{i}  =\begin{cases}
        \Theta  + \epsilon_i & i\in\mathcal{A},\\
        g(\Theta, \alpha) + \epsilon_i & i\in\mathcal{B},
    \end{cases}
    \qquad \text{ and } 
 E[\epsilon_i]  = \mathbf{0}.
\label{eq:fullmodel}
\end{equation}
$\epsilon_i$'s are symmetric matrices with zeroed diagonals, representing the joint residuals due to (a) the noise in estimating each correlation matrix, (b) the 
variation across units within the groups. When the working model for the data is not exact, (c) may also include the effects of modeling error; however, this error vanishes if the two group means are identical (that is, under the global null).

For the model to be useful, we need to identify link functions $g$ and corresponding parameter sets $(\boldsymbol{\Theta},A(\boldsymbol{\Theta}))  \subset \mathcal{C}_p\times \Omega_\alpha $ where the model behaves well. We look for the following properties: 
\begin{enumerate}
    \item \textbf{Produces a correlation matrix:}
    For any $(\Theta, \alpha)\in 
 (\boldsymbol{\Theta},A(\boldsymbol{\Theta}))$, $g(\Theta,\alpha)$ is symmetric, positive semi-definite with a unity diagonal (meaning $g(\Theta,\alpha) \in \mathcal{C}_p$).
    \item \textbf{Identifiable:}
    The model in \eqref{eq:fullmodel} is identifiable.
    \item \textbf{Invertible by} $\boldsymbol{\alpha}$\textbf{:}
    There exists $g^{-1}$ so that $g^{-1}\left(g\left(\Theta,\alpha\right),\alpha\right)=\Theta$.
    \item \textbf{Continuously differentiable:}
    The derivatives $\frac{\partial\vec{g_{j}}}{\partial\alpha_{k}}$ are defined and continuous.

\end{enumerate}
Properties 1 and 2 are required for the parametrization in \eqref{eq:fullmodel}, and Properties 3 and 4 are used for the estimation algorithm we propose. 
Whereas the rest of this paper treats the link function abstractly,
next we discuss our prototype link. 

\subsection{The Multiplicative Link Function}
\label{subsec:Multiplicative}
The multiplicative link \eqref{eq:MultLink}
is our prototype for $g$ throughout this paper. 
The multiplicative link function can be presented in matrix form as 
$g\left(\Theta,\alpha\right)
=\Theta\circ\alpha\alpha^{t}+I-D_{\alpha}^{2}$,
where $D_{\alpha}$ is a diagonal matrix with entries $\left[D_{\alpha}\right]_{jj}=\alpha_j$ and 0 otherwise.
The multiplicative link is useful for modeling differences between the groups, when the original correlation matrix is relatively dense. This link function incorporates two assumptions regarding the directions of change between populations: 
\begin{itemize}
    \item \emph{Assumption 1:} Researchers are interested in identifying variables where the correlation profile has changed directionally between two populations (increased or reduced correlation).
    \item \emph{Assumption 2:} The expected magnitude of the difference between correlation coefficients would be  proportional to the magnitude of the original correlation.
\end{itemize}
 We expect most $\alpha_j$'s to be near 1, where $\alpha_j=1$ indicates no directional change.

With respect to TGA, we may ask if there is a perturbation in the patient group that leads to reduced connectivity between one brain region and its neighbors compared to the healthy group. In that case,
we would expect the correlations involving that region  to decay toward 0.
We would further expect the strongest decay (in absolute terms) to be in pairs originally showing the strongest correlations.
In our model, these effects would be represented by having $\alpha_j < 1$.

\begin{remark}
    A simple covariance-level mechanism gives rise to the proposed multiplicative correlation model. Suppose that in the case population, only the region-specific variances change with respect to the control population, while the cross-covariances remain fixed. This would happen when the activity of region-specific independent sources increases or decreases between the two groups, whereas the synchronized sources do not change. Under this scenario, the expected population correlation matrices of the control and the case group ($\Lambda^A, \Lambda^B$) exactly follow the multiplicative model in \eqref{eq:MultLink}.  
    Explicitly, let $(\sigma^A_{jk})$ be the covariance matrix of population $A$ and assume without loss of generality that $\sigma_{jj} =1$. Let $c_j^B > -1$ be the change in variance for $j$'th variable in the case matrix, so $\sigma^B_{jk} = \sigma^A_{jk}$ if $j\neq k$ but $1+ c^B_j$ if $j=k$. Because the case correlation matrix is the scaled covariance, we find:
    \[ \rho^B_{jk}=\frac{\sigma^B_{jk}}{\sqrt{\sigma_{jj}}\sqrt{\sigma_{kk}}} = \frac{\sigma^B_{jk}}{\sqrt{1+c^b_j}\sqrt{1+c^b_k}}=\rho_{jk}^A \cdot \alpha_j \cdot \alpha_k, \qquad \alpha_j = \frac{1}{\sqrt{1+c^b_j}}.\]
    For $c_j^B\in (-1,0)$, we need to check that the case covariance is still positive definite; see condition 1 below. 
    
\end{remark}

\begin{remark}
The model can be seen as an adaptation of two-way models for table data to correlation matrices. To see this, consider the following two-way model for two groups of $p\times p$ data tables, where the effect of the case group is composed of separate additive or multiplicative row effects $\tau$ and column effects $\nu$: $E_{i\in \mathcal{A}}[Y^i_{jk}] = \mu^A_{jk}$, and $E_{i\in \mathcal{B}}[Y^i_{jk}] = \mu^A_{jk} \cdot \tau_j \cdot  \nu_k $. When we adapt this to correlation matrices, the symmetry means that row and column effects are the same, resulting in our model above. 
\end{remark}

The model with multiplicative link meets the requirements under the following conditions:

\begin{enumerate}
    \item \textbf{Produces a correlation matrix:} 
     Let  $\lambda_p$ be the minimal eigenvalue of $\Theta$, and $\alpha_{m},\alpha_{M}$ are the minimal and maximal values of $\left|\alpha_{j}\right|$, respectively. Then a sufficient condition for $g(\Theta,\alpha) \in  \mathcal{C}_p$ is that 
    $\lambda_{p}\cdot \alpha_{m}^{2}\geq\alpha_{M}^{2}-1$.
  The model is therefore well defined for $\alpha_M \leq 1$, and, in the typical case where $\lambda_p>0,\alpha_m > 0$ it is well defined for $\alpha_M$ in the neighborhood around 1.
  The proof is given in Appendix \ref{subsec:PSD-Proof}.
    
    \item \textbf{Identifiable:} 
    The multiplicative link model is identifiable
    if $\vec{\Theta}$ has sufficiently many non-zero entries.
    Concretely, define $e^{\left(j\right)}$ the elementary vectors of size $p$, $e^{\left(j,k\right)} = e^{\left(j\right)} + e^{\left(k\right)}$,
    and $\tilde{m}$ the number of non-zero entries in $\vec{\Theta}$.
    Let $A$ be the matrix of dimension $\tilde{m} \times p$ with $e^{\left(j,k\right)}$ in its rows,
    corresponding to $\vec{\Theta}$'s non-zero entries.
    Then, assuming $\alpha_j > 0$ for all $j$, the model \eqref{eq:fullmodel} with the multiplicative link is identifiable if $rank(A) = p$. The proof and the conditions for the case where $\alpha_j \geq 0$ are given in Appendix \ref{subsec:Model-Identifiability-Proof}.

    \item \textbf{Invertible by} $\boldsymbol{\alpha}$\textbf{:} If $\alpha_{j}>0\,$ for all $j$,
    the multiplicative link function's inverse can be expressed as $g^{-1}_{jk}\left(M,\alpha\right)=\frac{M_{jk}}{\alpha_{j}\alpha_{k}}\,\text{for}\,j\neq k$, and $1$ otherwise.

    \item \textbf{Continuously differentiable:} The multiplicative link function is differentiable
    for any $\alpha_{j}$:
    $\frac{\text{\ensuremath{\partial g_{jk}}}}{\partial\alpha_{\ell}}=\Theta_{jk}\left(\alpha_{k}1(\ell=j)+\alpha_j1(\ell=k)\right)$.
    
\end{enumerate}
We can conclude that for $\alpha$ values separated from 0, and a sufficiently dense correlation matrix $\Theta$, the model is well-defined. We find these conditions to hold well empirically for RS-fMRI datasets.

\subsection{Model Estimation}
\label{section:Model-Estimation}

Our target for estimation is the parameter vector of differences $\alpha$, whereas we treat $\Theta$ as a nuisance parameter.
To estimate $\alpha$ we use a generalized least squares formulation, with a fixed concentration matrix that does not depend on $\alpha$ nor $\Theta$. 

We define two \emph{inverse weight matrices}, $\invCa$ and $\invCb$, and use these for the quadratic loss function
$\mathcal{S}\left(\Theta,\alpha\right) 
=\mathcal{S}^{A}\left(\Theta\right)
+\mathcal{S}^{B}\left(\Theta,\alpha\right)$, where
\begin{align*}
    \mathcal{S}^{A}\left(\Theta\right) & =\sum_{i\in\mathcal{A}}\left(\vec{R_i}-\vec{\Theta}\right)^{t}\invCa\left(\vec{R_i}-\vec{\Theta}\right),\\
    \mathcal{S}^{B}\left(\Theta,\alpha\right) & =\sum_{i\in\mathcal{B}}\left(\vec{R_i}-\vec{g}\left(\Theta,\alpha\right)\right)^{t}\invCb\left(\vec{R_i}-\vec{g}\left(\Theta,\alpha\right)\right).
\end{align*}

For intermediate dimension $p$, optimizing the loss function $\mathcal{S}$ on $\Theta$ and $\alpha$ jointly is still computationally expensive due to the large model space. 
We therefore propose to optimize the parameters in an alternating manner:
at each step we optimize on $\alpha$ by gradient descent, and reestimate $\Theta$ by inverting the method-of-moments estimator. 
Formally, let $\hat{\Theta}_{k-1}$ be our estimate of $\Theta$ at the $k-1$ step, then the updates for $\hat{\alpha}_k,\hat{\Theta}_k$ at step $k$ are given by:
\begin{align}
    \hat{\alpha}_{k} & =\underset{\alpha}{\arg\min}\,\mathcal{S}^{B}\left(\hat{\Theta}_{k-1},\alpha\right), 
    \label{eq:alpha_k}
    \\
    \hat{\Theta}_{k} & =\left(n_{A}+n_{B}\right)^{-1}\left(\sum_{i\in\mathcal{A}}R_{i}+\sum_{i\in\mathcal{B}}g^{-1}\left(R_{i},\hat{\alpha}_{k}\right)\right).
    \label{eq:theta_k}
\end{align}
As initial values for the optimization, we set $\hat{\alpha}_0$ in Equation \ref{eq:alpha_k} to its null-value indicating $\Lambda^A = \Lambda^B$. 
For the multiplicative link function, this is $\hat{\alpha}_0 = \mathbf{1}$. 
We constrain the optimization so that $\Theta\in \mathcal{C}_p$. Note that while $S^A(\Theta)$ motivated the estimation procedure, it is not part of the loss. 

\subsubsection{Setting the inverse weight matrices}
\label{subsec:inverse_weight}

According to least-squares theory, the quadratic expression 
$\mathcal{S}^{B}$ is minimized by the true mean $\Lambda^B$, regardless of the choice of inverse weight matrices $\invCb$. Estimation is most efficient when $\invCb$ is set to the inverse covariance of the (vectorized) correlation matrix, with the covariance accounting for both the within-individual sampling variation and the between-individual variation. 

We use an approximation that captures the sampling variation, while keeping the matrices well-conditioned. 
For a  $p\times p$ covariance matrix $\Sigma$, we denote by $C(\Sigma)$ the $m\times m$ covariance of a vectorized and scaled $Wishart(\Sigma,1)$. 
(The explicit form of $C(\Sigma)$ is derived in Appendix \ref{appendix:correlMat}).
Then set 
$C_B=n_{B}^{-1}\sum_{i\in\mathcal{B}}C\left(R_{i}\right)$ to approximate the covariance due to sampling for the group, 
and $\invCb = (\lambda\cdot \bar{d} \cdot  Id_m + (1-\lambda) C_B)^{-1},$
where $\lambda\in [0,1]$ is a regularization parameter, and $\bar{d}$ is the average of the diagonal elements of $C_B$. 

We find that using $\invCb$ based on the regularized covariance (rather than diagonals) improves the estimation of the model parameters. We discuss the role of regularization in Appendix \ref{appendix:regularization}.
Note that the approximation still ignores the variance across subjects in the group, as well as the increased residuals due to modeling error. These will both be accounted for at inference by using the GEE (sandwich) estimator.

\begin{remark} 
In practice, the iterative optimization process defined by \eqref{eq:alpha_k}\eqref{eq:theta_k} is affected only by the choice of $C^*_B$ but not by $C^*_A$. $\alpha_k$ is directly estimated from the case samples (group B) through $C^*_B$. The updates of $\Theta_k$ do not depend directly on $C^*_A$ or $C^*_B$.
\end{remark}

\subsection{Model Inference}
\label{subsec:Model-Inference}

The estimation of $\alpha$ is viewed as a solution of a set of estimating equations: 
for a non-boundary minima of $\mathcal{S}$ by $\alpha$, the estimation error will be approximately Gaussian: 
$\hat{\alpha}-\alpha \overset{\cdot}{\sim}\mathcal{N}_{p}\left(\mathbf{0},V\left(\hat{\alpha}\right)\right)$, where $V\left(\hat{\alpha}\right)$ is the covariance matrix of our estimates. We propose two methods for estimating $V\left(\hat{\alpha}\right)$:
the sandwich estimator based on the generalized estimation equations (GEE) framework \citep{liang1986longitudinal}; and a two-sample Jackknife procedure that requires fewer assumptions but larger sample sizes and more computation.

\subsubsection{GEE Based Estimator for the Variance}

The covariance of the estimated parameters $\eta = (\vec{\Theta} ,\alpha)$ can be estimated using a sandwich estimator, following the general estimating equation framework of \cite{liang1986longitudinal}. For parameters that can be defined as solutions to an equation set, this is a robust estimator for the estimator covariance matrix. In particular, the covariance estimator should be robust to a mismatch between the working inverse covariance $C^*_B$ and the true one. 

The joint estimating function is $U(\eta) = (U_{\vec{\Theta}}(\eta)^t, U_\alpha(\eta)^t)^t $, with 
\[U_{\vec{\Theta}}(\eta) = \sum_{i\in \mathcal{A}} (\vec{R}_i - \vec{\Theta}) + \sum_{i\in \mathcal{B}} (g^{-1}(\vec{R}_i ,\alpha) - \vec{\Theta}), \qquad
U_\alpha(\eta) = J_\alpha(\eta)^t C^*_B \sum_{i\in \mathcal{B}} (\vec{R}_i - g(\vec{\Theta}, \alpha)), \]
where $J_\alpha(\eta)$ denotes the Jacobian of $g(\vec{\Theta}, \alpha)$ with respect to $\alpha$, and $C^*_B$ the working inverse covariance used in the optimization of group B. Our estimator $\hat{\eta}$ is the solution to $U(\eta) = 0$.
Furthermore, for each subject, let $u_i(\eta)$ denote its contribution to the joint estimating function. For subjects in group A, 
$u_i(\eta) = ((\vec{R}_i - \vec{\Theta})^t, \mathbf{0}^t)^t$. 
For subjects in group B, 
$u_i(\eta) = \left(g^{-1}(\vec{R}_i,\alpha) - \vec{\Theta}),\,\, J_{\alpha}(\eta)^t C^*_B \left(\vec{R}_i - g(\vec{\Theta}, \alpha)\right)\right)^t.$

The sandwich covariance estimator for $\hat{\eta}$ is defined as $\hat{V}_{gee}(\eta) = I^{-1} \Upsilon I^{-t}$, where the ``bread matrix'' is the derivative of the estimating equations with respect to the parameter evaluated at $\eta = \hat{\eta}$,
$I = \frac{\partial U(\eta)}{\partial \eta^t}$; and the middle matrix empirical covariance of the residual vectors 
$$\Upsilon = \frac{n_A}{n_A - 1}\sum_{i \in \mathcal{A}}u_i(\hat{\eta})u_i(\hat{\eta})^t+ \frac{n_B}{n_B-1}\sum_{i\in \mathcal{B}}u_i(\hat{\eta})u_i(\hat{\eta})^t$$. The estimated covariance matrix of $\hat{\alpha}$ is the lower right $p\times p$ block of $\hat{V}_{gee}(\eta)$.
In practice, for the multiplicative link, the covariance of $\eta$ does not need to be explicitly constructed; instead we can construct $\hat{V}_{gee}(\hat{\alpha})$ using a Schur decomposition of the blocks. See Appendix \ref{appendix:block_gee}.

The number of samples and the temporal correlation scale the covariance of each correlation matrix. In practice, because we assume that all correlation matrices were estimated from similar processes, these constants are equal, and they cancel out in the sandwich estimator. We therefore ignore them.

\subsubsection{Jackknife Based Estimator for the Variance}

Alternatively, we propose a non-parametric method of estimating $V\left(\hat{\alpha}\right)$
with a Jackknife procedure \citep{efron1981jackknife}.
This framework can be used in cases where the GEE framework underestimates the variance, for example when the proportion of the control group is small.
Here, each subject is omitted in turn from the sample, and $\alpha$ is re-estimated.
Call $\hat{\alpha}_{(-i)}$ the estimate calculated when omitting subject $i$,
and call $\hat{\alpha}_{A}=n_A^{-1}\sum_{i\in\mathcal{A}}\hat{\alpha}_{(-i)}$
($\hat{\alpha}_{B}$ defined similarly).
Then we estimate the centering quantity
$    \hat{\alpha}_{jack}=
    \frac{n_{A}\cdot\hat{\alpha}_{A}+n_{B}\cdot\hat{\alpha}_{B}}{n_{A}+n_{B}}.  $
The variance of the Jackknife estimate is estimated with $\hat{V}_{jack}\left(\hat{\alpha}\right)=\nu_A + \nu_B$ where $    \nu_{A}=\frac{n_A - 1}{n_A}
    \sum_{i\in\mathcal{A}}
    \left(\hat{\alpha}_{(-i)}-\hat{\alpha}_{A}\right)\cdot
    \left(\hat{\alpha}_{(-i)}-\hat{\alpha}_{A}\right)^t $ and $\nu_{B}$ is defined similarly.
We use a warm start for $\hat{\alpha}_{(-i)}$ for faster convergence.

\subsubsection{Inference on $\alpha$ and Accounting for Multiplicity}
After estimating $\hat{V}\left(\hat{\alpha}\right)$, we use Wald tests for region inferences. 
To reduce the sensitivity of the test statistics to occasional downward estimation errors in the parameter variances \citep{wang2021moderated}, we add a small positive constant $s_0$ to each estimated standard deviation
\[z_j = \frac{\hat{\alpha}_j - \alpha_j}{s_j + s_0},\qquad s_j = \sqrt{\hat{V}_{jj}},\quad s_0 = c\cdot median(s_j),\]
following the stabilization approach used in large-scale genomic inference \citep{tusher2001significance}.  
We then assume $z_j \overset{\cdot}{\sim} \mathcal{N}(0,1)$ to calculate the Z-values, p-values, and confidence intervals accordingly.
The BH procedure is used to correct the p-values for multiplicity \citep{benjamini1995controlling}, and the selected confidence intervals are corrected for selection using the BY procedure for false coverage rate (FCR) control \citep{benjamini2005false}.

\section{Results}

\subsection{Simulations}
\label{section:simulations}

In this section, we evaluate the model on simulated data.
The data is simulated in accordance with the model's assumptions with the multiplicative link function.
In each experiment, we randomly generate a dense spatial correlation matrix $\Theta$.
We then sample for each subject $T$ multivariate Gaussian measurements from the group covariance ($\Theta$ for control subjects
and $g\left(\Theta,\alpha\right)$ for case subjects), which we summarize into a single sample correlation matrix per subject.
The measurements are simulated either from a multivariate $ARMA\left(1,1\right)$ process or from a process with temporal independence. For the ARMA process, both auto-regressive and moving-average parameters are equal to 0.5 to approximate the temporal correlation expected in RS-fMRI scans \citep{luo2020improved, locascio1997time}. The regularization constant $\lambda = 0.1$ is used; in inference, the variance stability parameter is $c=0.03$ (results are robust to these two choices). 
More details about the simulations are in Appendix \ref{section:sampling_algo}.
Results on the algorithm behaviour when data is not sampled according to the multiplicative link are found in Appendix \ref{appendix:misspecified} 

\subsection{Validation of the Inferential Methods}

\begin{figure}[tbh]
\begin{center}
    \includegraphics[scale=0.75]{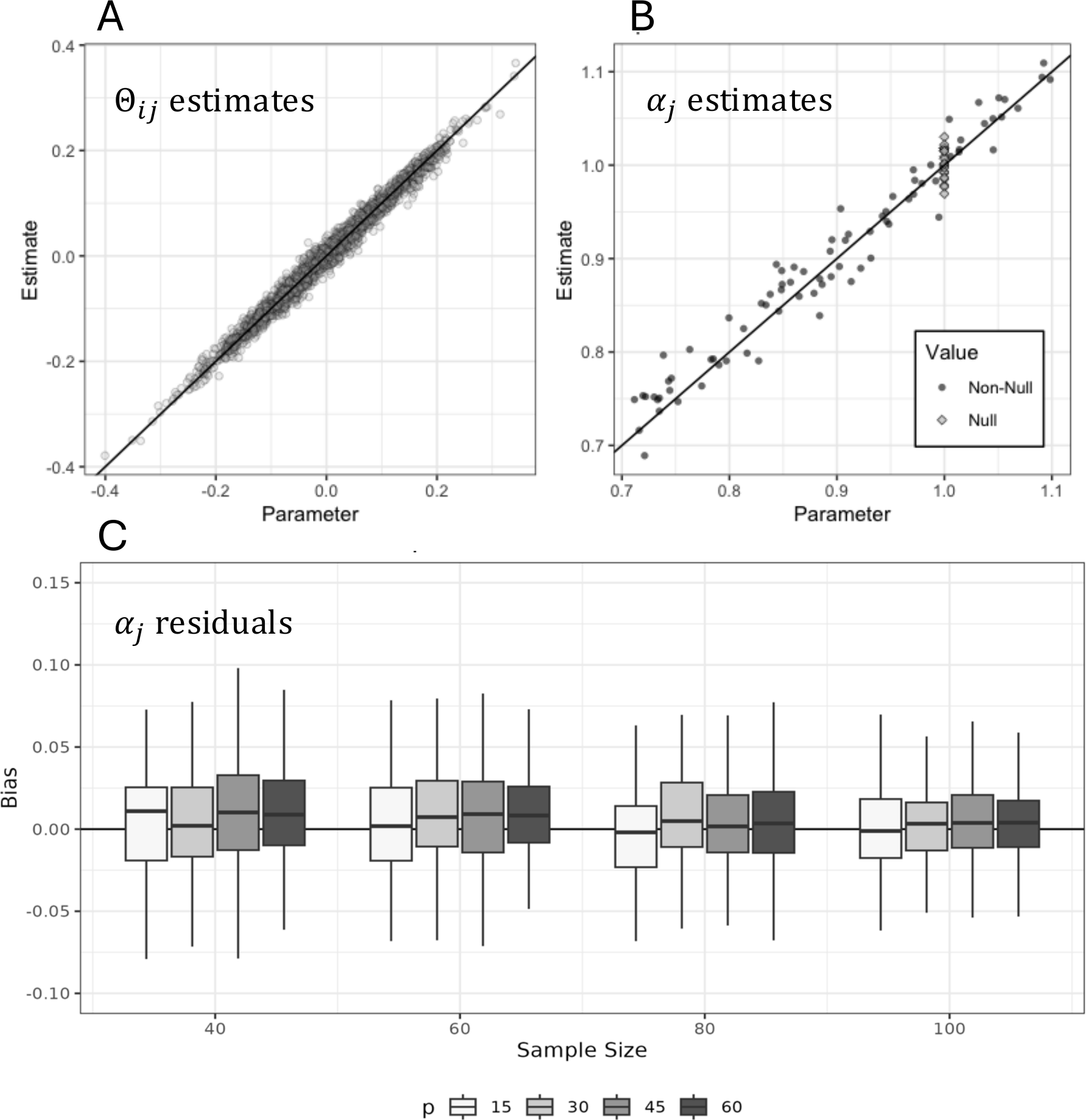}
    \caption{ \small  
    \textbf{Parameter recovery in simulations of the multiplicative correlation-difference model.}
    Estimated versus true parameter values: A. Baseline correlation coefficients ($\Theta_{jk}$). B. Variable perturbation effects ($\alpha_{j}$) in a single simulation. 
    Points near the identity line indicate accurate recovery of correlation structure or perturbation effects.
    C. Estimation errors (residuals) for perturbation effects $\alpha_j$s across repeated simulation runs, shown as a function of sample size ($n$) and matrix dimension ($p$). 
    Errors decrease as $n$ and $p$ increase.}
    \label{fig:bias}
\end{center}
\end{figure}
First, to evaluate parameter recovery, we compare the estimated parameters to their true values under the multiplicative model. 
We find that for both the baseline correlation coefficients $\Theta_{ij}$ (Figure \ref{fig:bias}A) and the variable perturbation effects $\alpha_j$ (Figure \ref{fig:bias}B) the estimated values align well with the true parameters, as indicated by the identity line. The panels show a single simulation, where we set 10 $\alpha_j$'s equal to 1, reflecting no change between the two groups (\emph{null}s), and let the rest vary between $\left[0.7, 1.1\right]$.

We then examine how estimation accuracy of the variable effects depends on sample size ($n$) and matrix dimension ($p$, Figure \ref{fig:bias}C). Both the bias and the variability of $\hat{\alpha}$ decrease as sample size increases, and stability increases with matrix dimension as well. For the sample sizes discussed, the observed distribution for individual $\hat{\alpha_j}$  is symmetric and well approximated by a Gaussian. 

We next evaluate the accuracy of the model-based standard deviation estimates. 
For each randomly generated set of $(\Theta,\alpha)$, we simulated $B=100$ datasets and computed the empirical standard deviation of each $\hat{\alpha}_j$ across runs.
We compared these empirical estimates to the average model (GEE) standard deviation estimates over these runs (Figure  \ref{fig:gee-var-consistency-1}A). Here, each point represents the standard deviation of estimating one $\hat{\alpha}_j$, estimated by comparing multiple estimates (x-axis) and by averaging multiple GEE model estimates (y-axis). Across simulation settings, we observe substantial heterogeneity in the observed standard deviations of $\alpha_j$ estimates. The GEE model estimates capture well this heterogeneity. Most estimates are along the unity line, indicating well-calibrated estimates. We further visualize the percentage error in estimated standard deviation across simulation settings (Figure \ref{fig:gee-var-consistency-1}B). The distribution of errors shows a subtle downward bias ($\sim1\%$), though estimation noise is non-negligible.  
We repeat the experiment using the Jackknife estimator in Supplementary Figure \ref{fig:jack-var},
obtaining similar results. Note that results for the jackknife are based on a smaller number of repeats so are not directly comparable. 

\begin{figure}[!tbh]
\begin{center}
    \includegraphics[scale=0.8]{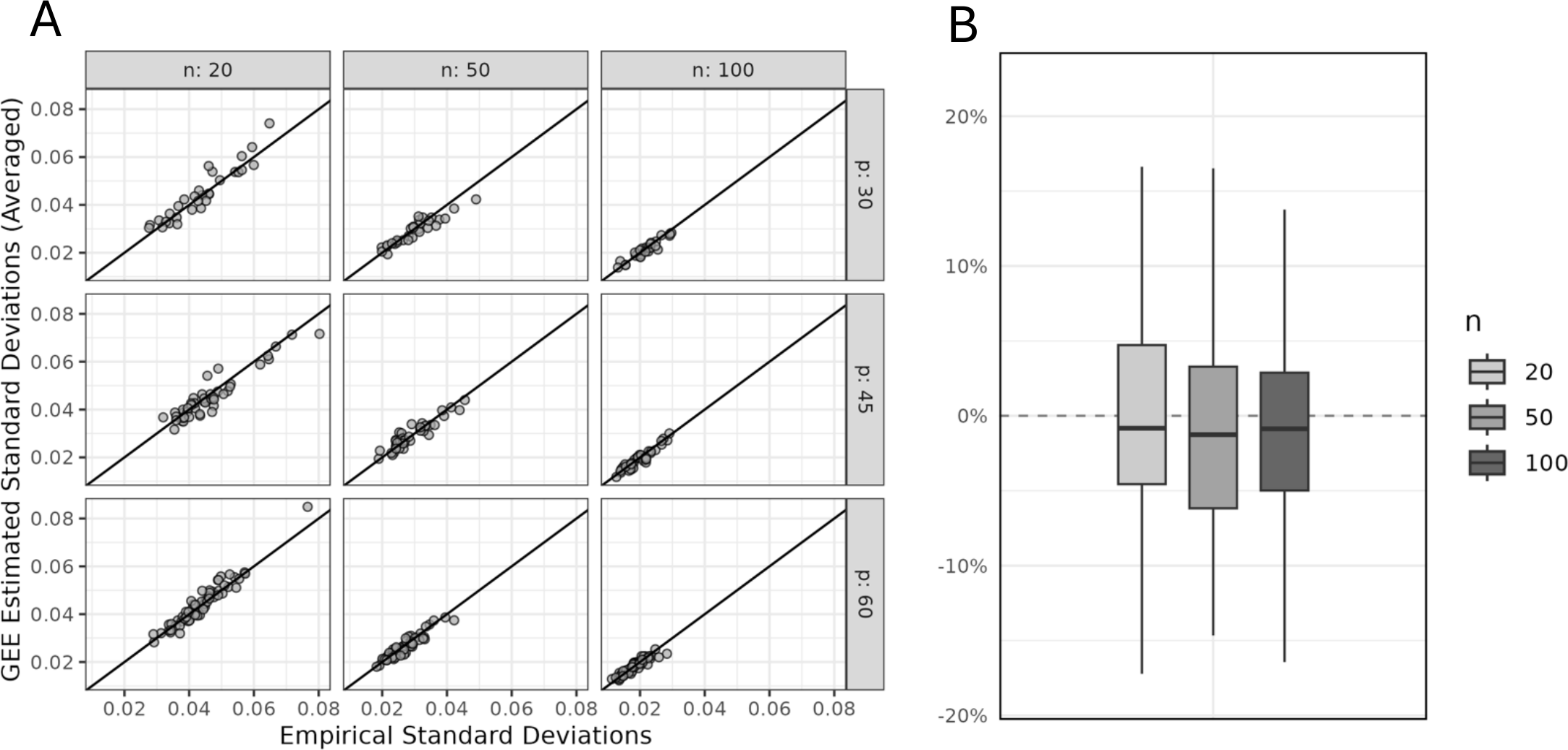}
    \caption{ \small
    \textbf{Accuracy of standard deviation estimates for the perturbation parameters $\hat{\alpha}$}.
        A. Comparison of empirical standard deviations obtained from repeated simulations and model-based standard deviation estimates derived from the GEE framework. Each point compares the empirical SD across B estimates with mean GEE-estimated SD. Panels show different combinations of sample size ($n$) and matrix dimension ($p$). The close alignment with the identity line indicates that the GEE estimator captures heterogeneity in uncertainty across variables. B. Boxplots of the percentage error in estimated standard deviations across simulation settings, split by sample size. A slight downward bias is observed. Simulations were conducted under the weak null with positive temporal autocorrelation.}
    \label{fig:gee-var-consistency-1}
\end{center}
\end{figure}

\begin{figure}[!tbh]
\begin{center}
    \includegraphics[scale=0.8]{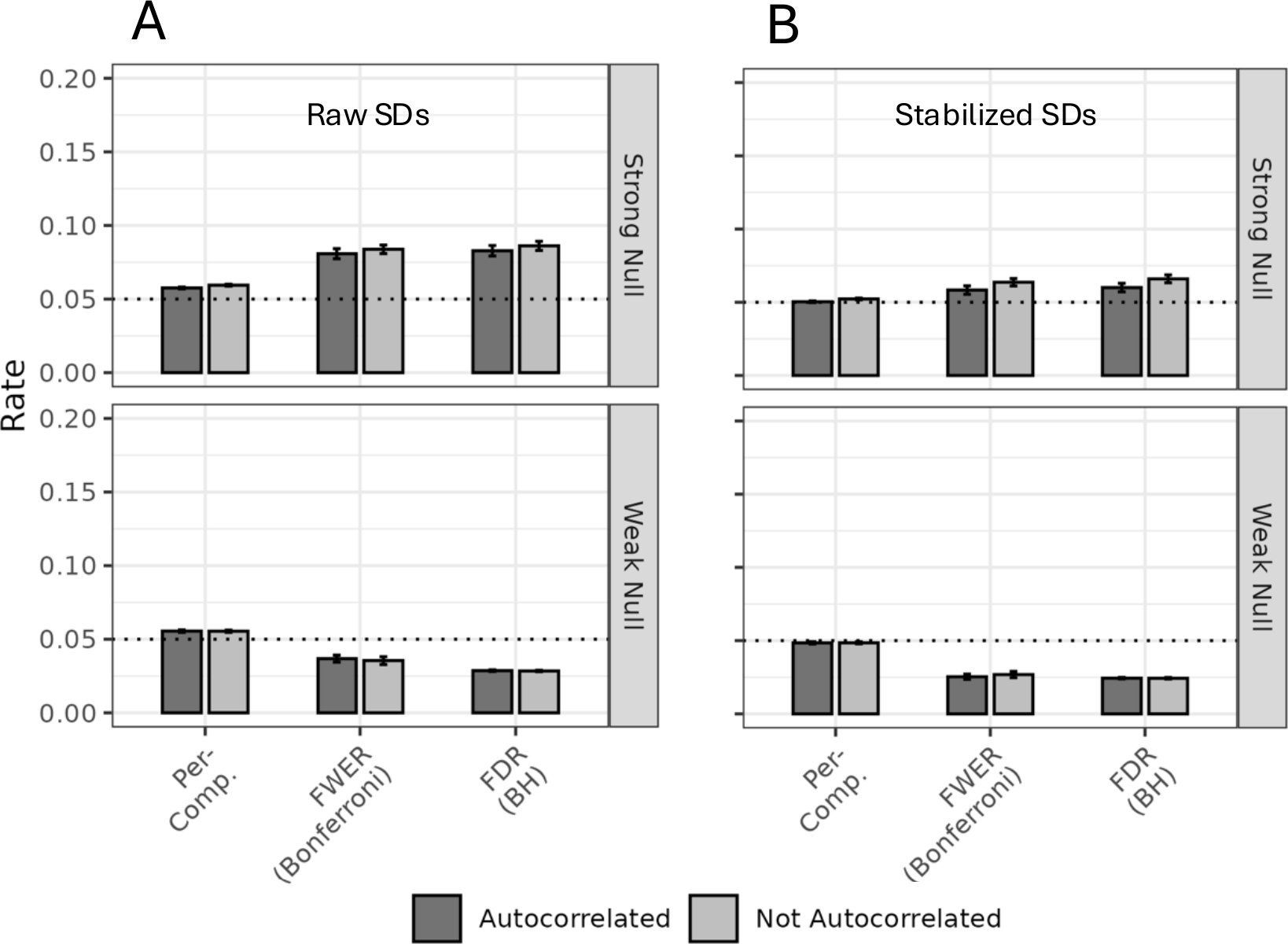}
    \caption{\small
    \textbf{Control of Type I error under partial hypotheses testing.}
    Per comparison error rate, family-wise error rate (FWER; Bonferroni correction), and false discovery rate (FDR; BH correction) are shown. Panel A: Without stabilization, the procedure maintains per-comparison error rate at the nominal 0.05 level, and FDR or FWER at a higher 0.08 rate under the strong null. Under the weak null all methods are calibrated. Panel B: Stabilizing the standard deviations ($c=0.03$) results in near-nominal errors for the strong null as well. Error bars indicate $\pm 1$ standard errors across different $(\vec{\Theta}, \alpha)$ simulation scenarios. All simulations with $p=32, n=100$. 
    }
    \label{fig:fwer-fdr}
\end{center}
\end{figure}
\begin{figure}[!bth]
    \begin{center}
    \includegraphics[scale=1.2]{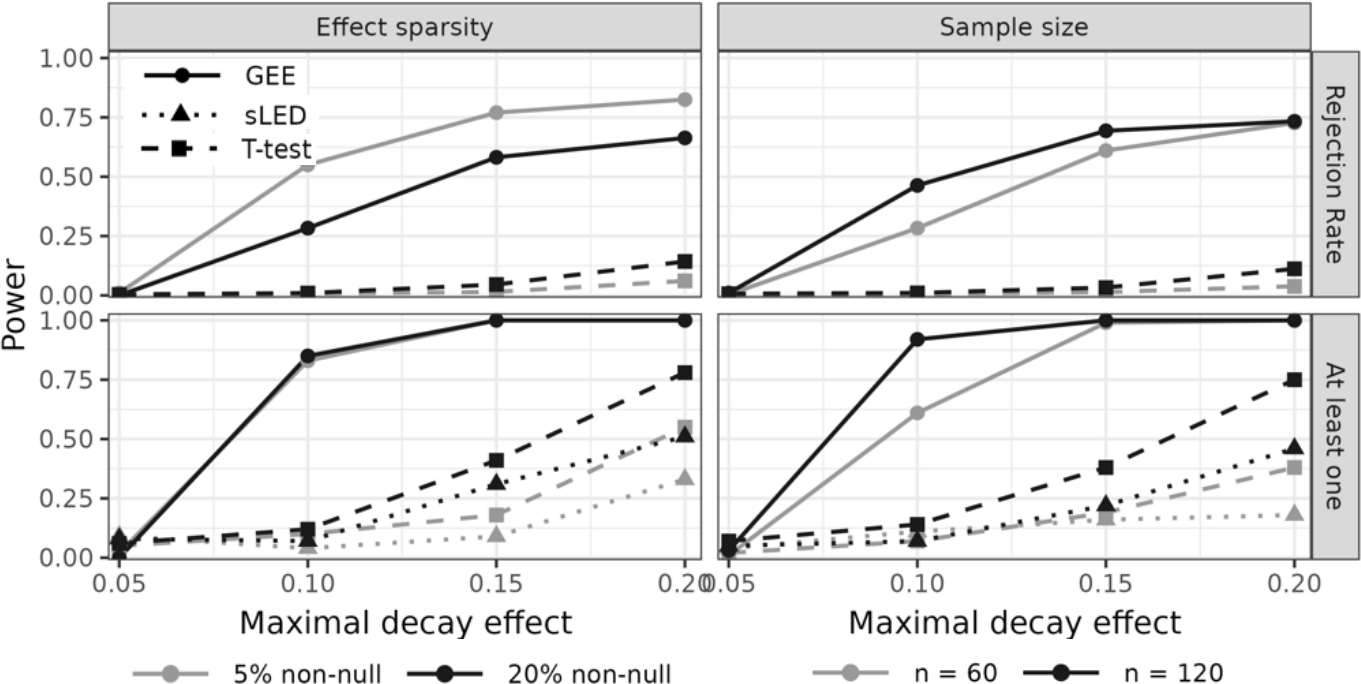}
    \caption{\small
    \textbf{Improved power of the proposed model compared to alternative methods}.
    Statistical power is evaluated as a function of maximal effect size $\max_j(1 - \alpha_j)$. The proposed method (solid line, $\CIRCLE$) demonstrates higher sensitivity than mass-univariate Welch t-tests (dashed line, $\blacksquare$) and sparse leading eigenvalue method (sLED, dotted line, $\blacktriangle$) particularly when effects are sparse or sample sizes are moderate. The top row shows proportion of truly non-null hypothesis rejected. The bottom row the probability to reject the global null. Left panels vary the proportion of perturbed variables while holding total sample size constant ($n = 100$). Right panels vary total sample size while holding the proportion of perturbed variables fixed at 10\%. Colors denote either the proportion of non-null $\alpha_{j}$-s (left panels) or total sample size (right panels).}
    \label{fig:power}
\end{center}
\end{figure}
Finally, we evaluate the control of the Type 1 error under multiplicity corrections (Figure \ref{fig:fwer-fdr}). Error rates are examined under both the strong null (all effects equal 1, top) and the weak null (mixture of null and non-null effects, bottom), with or without temporal autocorrelation. When no nugget standard-deviation is used (Panel A, $c=0$), normal-distribution based statistical tests with the GEE standard deviation estimates (\ref{fig:fwer-fdr}) finds per comparison rates remain near the nominal 0.05 level. However, for the strong-null condition (all $\alpha_j=1$) with Bonferroni or BH correction, we observe inflated rates of 0.08. Using a small nugget (here $c=0.03$) reduces these rates to near nominal level.

\subsection{Comparing Power to Alternative Methods}

We next compared the statistical power of the proposed model to two alternative methods: a mass-univariate approach based on two-sample t-tests with BH correction on the Fisher-transformed correlation coefficients (Figure \ref{fig:power}), and the sparse leading eigenvalue (sLED) method of Zhu et al. (2017), which tests for global differences between correlation matrices via a permutation framework.
For each panel, the bottom rows display the probability of rejecting the global null hypothesis, meaning that we found significant $\hat{\alpha}_j \neq 1$ for at least one variable after multiplicity corrections. Methods are discriminated by line type and marker type, whereas design condition is discriminated by color. Because sLED provides inference only for the global null, we use it only to compare global-null rejections. The proposed method demonstrates substantially higher sensitivity across a range of maximal decay magnitudes ($\max_j(1-\alpha_j)$).The top row shows the proportion of truly non-null variables correctly identified. Across effect sizes, the proposed model consistently achieves higher variable-level detection rates, particularly when effects are sparse (panel A) or when sample sizes are moderate (panel B). 
The gain in power is most pronounced when the proportion of perturbed variables is small, reflecting the advantage of imposing structured column-wise perturbations rather than testing each correlation coefficient independently. Conversely, when the number of columns affected increases, the proportion of rejections in our model is reduced (top left). We hypothesize that as more columns are affected, it is harder for the model to identify the columns.
Results regarding the behavior of the procedure when the deviations do not follow the multiplicative decay model are shown in Table \ref{tab:power-comparison} and discussed in Appendix \ref{appendix:misspecified}. 

\subsubsection{Subject-level variability}
The GEE variance estimators should be able to accommodate within-group subject variability and to correctly estimate its effect on the estimator. To test this property, we experiment with increasing subject-level variability, while keeping the population-level means stable per group. We use two designs: (a) adding noise to the subject correlation coefficients, which modestly increases estimator variance; and (b) adding per-subject multiplicative random effect, which directly affects the estimator variance. The exact formulation is given in Appendix \ref{appendix:subjectvar}. In Figure \ref{fig:sub_variability} we compare average empirical variance of the estimators with average estimated variance across repeats and parameters, under strong and weak null conditions. We find that the estimated GEE variance tracks well the increase in observed empirical variance. See Appendix \ref{appendix:subjectvar} for the experimental details. 
\begin{figure}[tb]
\begin{center}
    \includegraphics[scale=0.65]{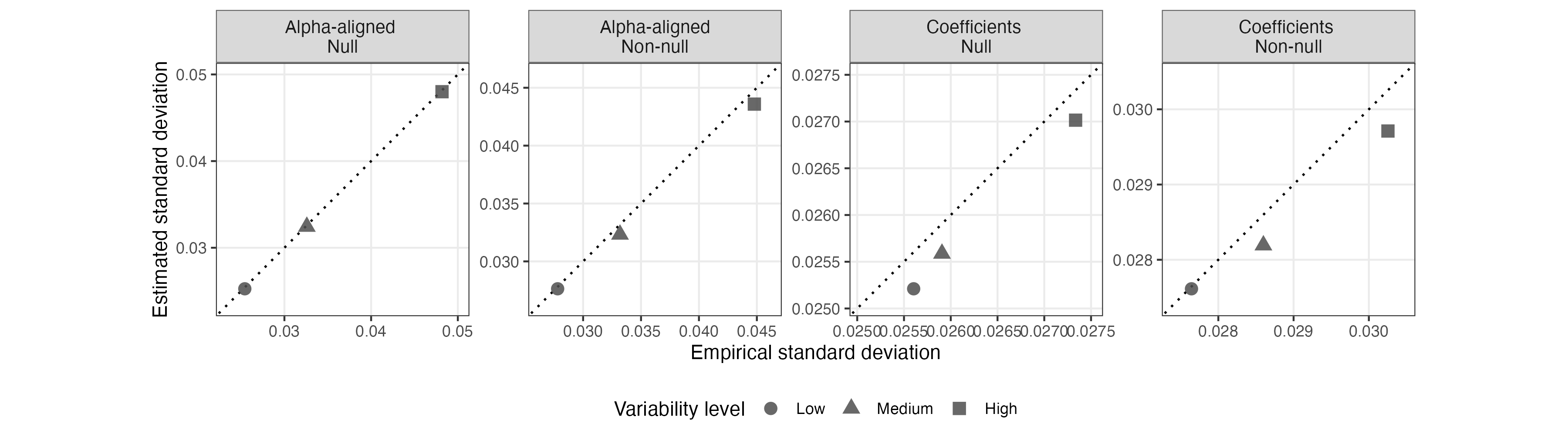}
    \caption{\small
    \textbf{GEE estimator accommodates effects of subject-specific variability.} Plots compare average empirical estimator SD (x-axis) with average GEE-estimated SD (y-axis) under three levels of inter-subject variability. In all plots we see that empirical and estimated SDs increase in tandem, meaning that the estimator can identify the effects of subject specific variations. The left panels correspond to random multiplicative perturbations, whereas the right panels correspond to correlation-coefficient perturbations. }
    \label{fig:sub_variability}
\end{center}
\end{figure}

\subsection{Application to Data From Patients with TGA}
\label{section:realdata}
We refer to the RS-fMRI data recorded by \cite{peer2014reversible} from patients with acute transient global amnesia (TGA) to demonstrate our proposed model.
In the original paper, the data was analyzed using the mass-univariate approach.
In order to preserve statistical power,
the authors focused their investigation on suspected brain regions and
they grouped together activity across known sets of brain regions to limit the number of comparisons.
While these practices could detect statistically significant differences between the groups,
they limited the ability to detect differences in brain regions that the researchers had no prior hypothesis about.

\begin{figure}[!tbh]
\begin{center}
    \includegraphics[scale=0.7]{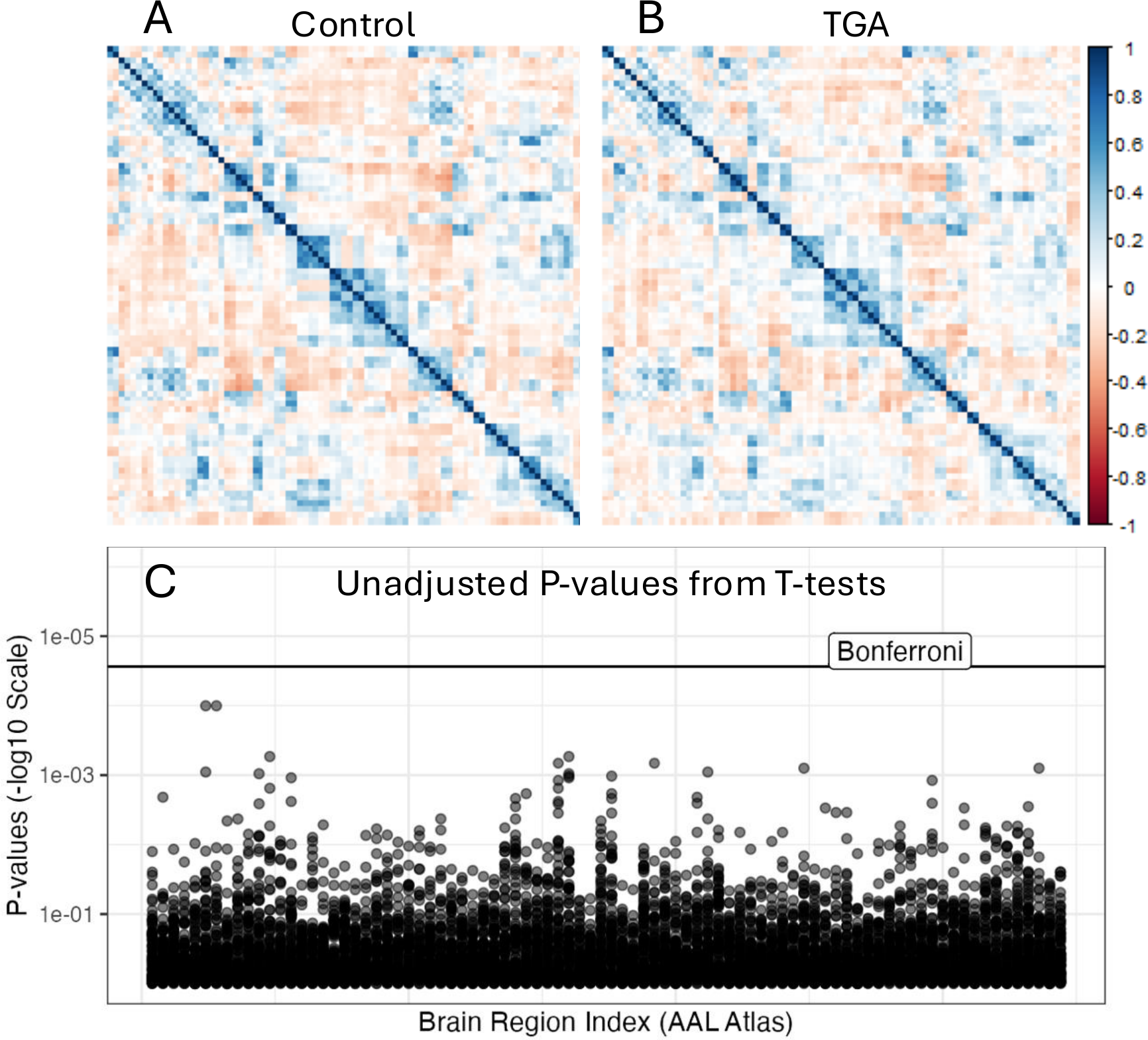}
    \caption{\small 
        \textbf{Group-average correlation matrices in control subjects and TGA patients}.
    Mean resting-state functional connectivity matrices computed across subjects in the control group (panel A) and the TGA group (panel B), based on AAL parcellation ($p = 86$ regions).
    C. Manhattan plot for p-values computed using two-sample t-tests for the correlation coefficients of each region pair. The Bonferroni / BH line is shown. No adjusted pairwise correlation is significant after BH adjustment; minimal adjusted p-value $= 0.37$.}
    \label{fig:tga-exploratory}
    \end{center}
    \end{figure}

\begin{figure}[!tbh]
\begin{center}
    \includegraphics[scale=0.7]{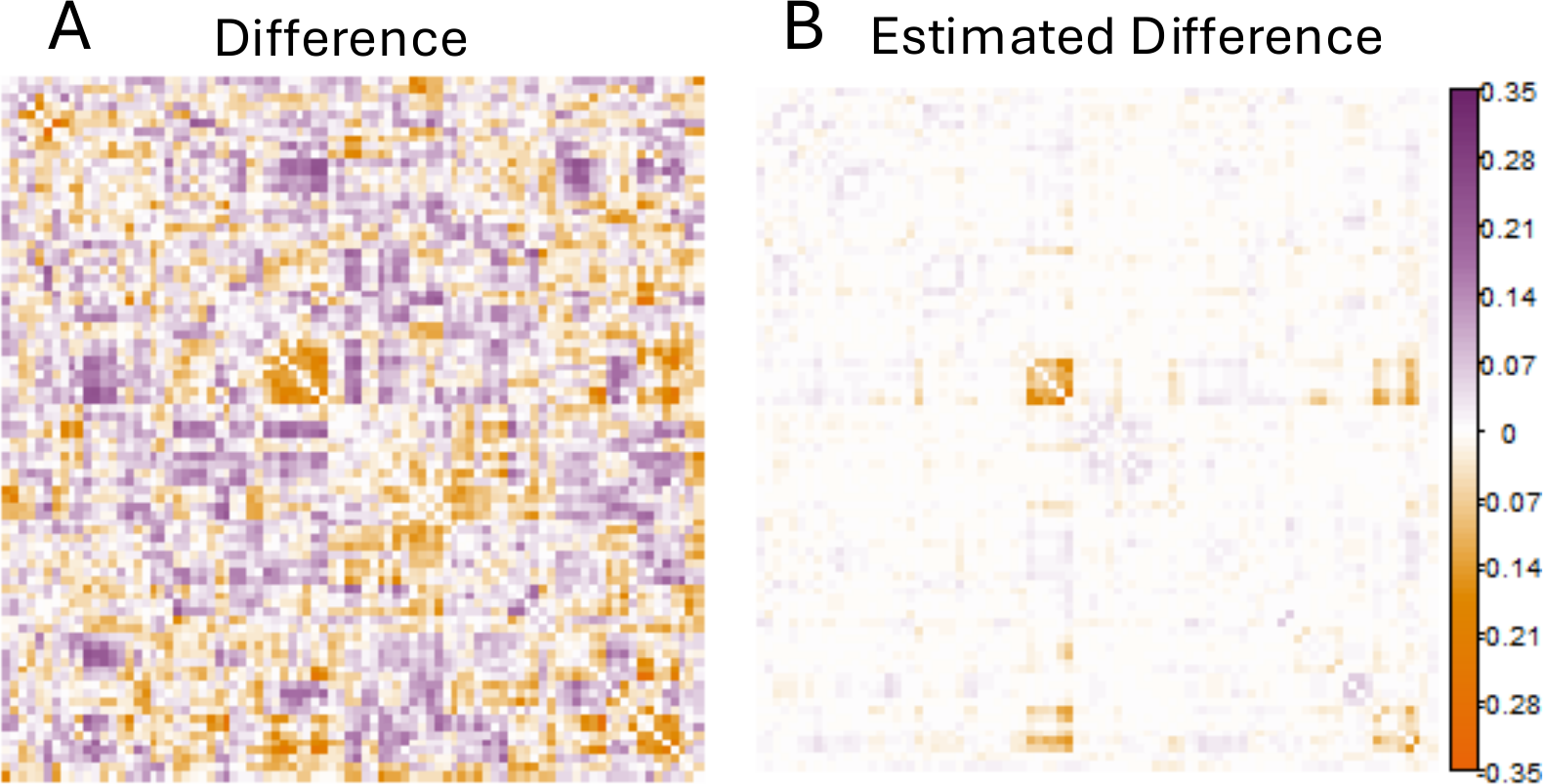}
    \caption{\small
    \textbf{Observed versus model-estimated differences in correlation structure between groups.}
    A. Empirical difference between group-average correlation matrices.
    B. Differences implied by the fitted multiplicative model.
    The model captures structured, region-specific directional changes while imposing a column-wise perturbation constraint, resulting in smoother and more parsimonious estimates compared to the raw difference matrix.}

    \label{fig:tga-difference}
\end{center}
\end{figure}

The data consist of $n_B=14$ patients (aged $62.7\pm7.4$; 8 male) that met the standard clinical criteria for a diagnosis of TGA,
matched by $n_A=23$ volunteers (aged $62.1\pm6.9$; 11 male) that serve as the control group.
Patients were scanned with functional MRI during the few hours of the acute amnestic state. The rarity of TGA and the time constraints explain the small sample sizes in this study. We refer the reader to the original paper for details on the fMRI processing and extraction of activity patterns in each region.
The whole brain network was defined according to the Automatic Anatomical Labeling (AAL) atlas,
which defines 45 brain regions in each cerebral hemisphere \citep{tzourio2002automated}.
Four regions were omitted from the analysis
as some subjects had missing values in the corresponding columns,
resulting in correlation matrices of dimension $p=86$.

We begin by examining the group-average correlation matrices for the control and TGA groups  (Figure \ref{fig:tga-exploratory}). Both matrices exhibit a dense connectivity structure with broadly similar patterns across regions, and no visually striking large-scale disruptions are apparent.

To assess pairwise differences, we applied a mass-univariate approach, conducting two-sample t-tests for each off-diagonal correlation coefficient ($m=3,655$), followed by BH correction. The resulting p-values are shown in the Manhattan plot in Figure \ref{fig:tga-exploratory}C. After correction for multiple comparisons, no correlation pairs remained significant (minimum adjusted p-value = 0.37). Thus, without imposing structural constraints on the difference between groups, a standard mass-univariate analysis would fail to detect any significant connectivity differences. 

We next estimated our model on the data,  producing a vector of estimates $\alpha$ and their associated standard deviations $(\sqrt{\hat{V}_{jj}})$ using the GEE method. The identifiability condition in Section \ref{subsec:Multiplicative} was verified by constructing a set of edges for which $\rho^A_{jk} = 0$ can be rejected with family-wise error  $\leq0.05$, resulting in a full rank matrix as required. The p-values and intervals were corrected to control the FDR and FCR respectively \citep{ benjamini2005false}.

\begin{table}[ht]
\centering
\begin{tabular}{clccccc}
  \hline
Index & AAL Region & $\hat{\alpha} \,\,\,$& (sd) & Direction & p-value & $p_{adj}$ (BH) \\ 
  \hline
 \textbf{42}& \textbf{Amygdala (R)} & \textbf{0.78} &\textbf{(0.06)} & - &\textbf{ 0.00028 }& \textbf{0.024} \\ 
  41 & Amygdala (L)& 0.84& (0.05) & - & 0.00167 & 0.071 \\ 
  37 & Hippocampus (L) & 0.89 &(0.03) & - & 0.003 & 0.082 \\ 
  83 & Temporal Pole: Sup. Temporal Gyrus (L) & 0.90 & (0.03) & - & 0.006 & 0.14 \\ 
 87 & Temporal Pole: Mid. Temporal Gyrus (L) & 0.84 &(0.06) & - & 0.015 & 0.26 \\ 
 23 & Superior Frontal Gyrus, medial (L) & 1.05 &(0.02) & + & 0.023 & 0.33 \\ 
  50 & Superior Occipital Gyrus (R) & 1.06&  (0.03) & + & 0.033 & 0.41 \\ 
 90 & Inferior Temporal Gyrus (R) & 0.91 &(0.04) & -& 0.05 & 0.41 \\ 
 78 & Thalamus (R) & 0.92 &(0.04) & - & 0.05 & 0.49 \\ 
 \hline
\end{tabular}
\caption{\small{Estimated region effects for the TGA data. Regions shown have nominal (uncorrected) p-values less than 0.05. In bold, only amygdala (R) was found significant at 0.05 level after multiplicity correction. Indices reflect the original AAL atlas.}}
\label{table:full-tga-results}
\end{table}

In Figure \ref{fig:tga-difference}, we show the observed difference ($\bar{R}_{B}-\bar{R}_{A}$, Panel A) and the estimated difference through the link function (Panel B). We see that the model captures directional decay in correlations for the patient group (orange), predominantly in the amygdala and hippocampal areas. Our methodology underestimates the differences in comparison to the mass univariate approach, due to the row-and-column-wise structure assumed in the model; however it identifies a strong pattern in the decay. 

These differences are further reflected in the estimated effects shown in Figure \ref{fig:tga-model-estimates}A. We see $\alpha_j$ estimates less than 1 for multiple regions, including the amygdala (R,L) the hippocampus (L) and the regions in the left temporal pole; we see increases in the superior frontal gyrus and the superior occipital gyrus.  After correction for multiplicity, the decay in  region 42, the right amygdala, is found to be statistically significant ($\hat{\alpha} = 0.78,\,\text{two sided } p_{adj} = 0.024$). The two next regions that had small p-values but did not cross 0.05 after adjustment are the left amygdala ($\hat{\alpha} = 0.84, \,p_{adj}= 0.071$) and the left hippocampus ($\hat{\alpha} = 0.89, \,p_{adj}= 0.082$). Regions with unadjusted p-values below 0.05 are shown in Table \ref{table:full-tga-results}, and all results are plotted in Figure \ref{fig:tga-model-estimates}. 

We ran two resampling-based procedures to validate the results. In the first, we used the bootstrap algorithm to check the validity of the standard deviation estimates on the TGA-data. In Figure \ref{figure:bootstrap} (Appendix \ref{appendix:bootstrap_tga}) we evaluate the bias and the standard deviation estimates against bootstrap samples. We compare the standard deviation estimated using GEE $\tilde{s}_j = s_j + s_0$ to with the empirical standard deviations across bootstrap samples $\hat{\sigma}_j$. The estimates are similar, with a correlation of $r=0.98$ and an average relative difference ($\delta = avg\{|\hat{\sigma}_j - \tilde{s}_j|/\hat{\sigma}_j\}$) of $ \delta= 8.4 \%$.  Second, to validate the p-values independently of the SD estimates, we rotated the patient matrices to remove the between-group difference and used $max_j\{t(\hat{\alpha_j})\} = \max_j|1-\hat\alpha_j|$ as a test statistic for a permutation test. The procedure rejected the null hypothesis for the right amygdala at $p_{adj} = 0.03$.  Both procedures are described in detail in Appendix \ref{appendix:bootstrap_tga}. 

The apparent involvement of the hippocampus in the functional data is consistent with lesions found in these regions in structural MRI. Notably, the unilateral implication is also in line with the structural findings which are mostly ($\sim90\%$) unilateral \citep{bartsch2006selective,Bartsch2010}. The involvement of the amygdala \citep{Segobin2023} is more interesting. Neuroanatomically speaking, the amygdala is functionally conjugated to the hippocampus giving rise to a unified hippocampal-amygdala circuit. This involvement may explain the frequency of anxiety in TGA \citep{Dohring2014}.

In summary, the application of our method on the TGA data identified a significant decay in one region (and borderline significant decay in two more regions) considered related to the processing of short- and long-term memory.
By introducing structure, our model had fewer parameters and a reduced number of hypotheses,
and was able to detect perturbed regions.

\begin{figure}[tb]
    \begin{center}
    \includegraphics[scale=0.7]{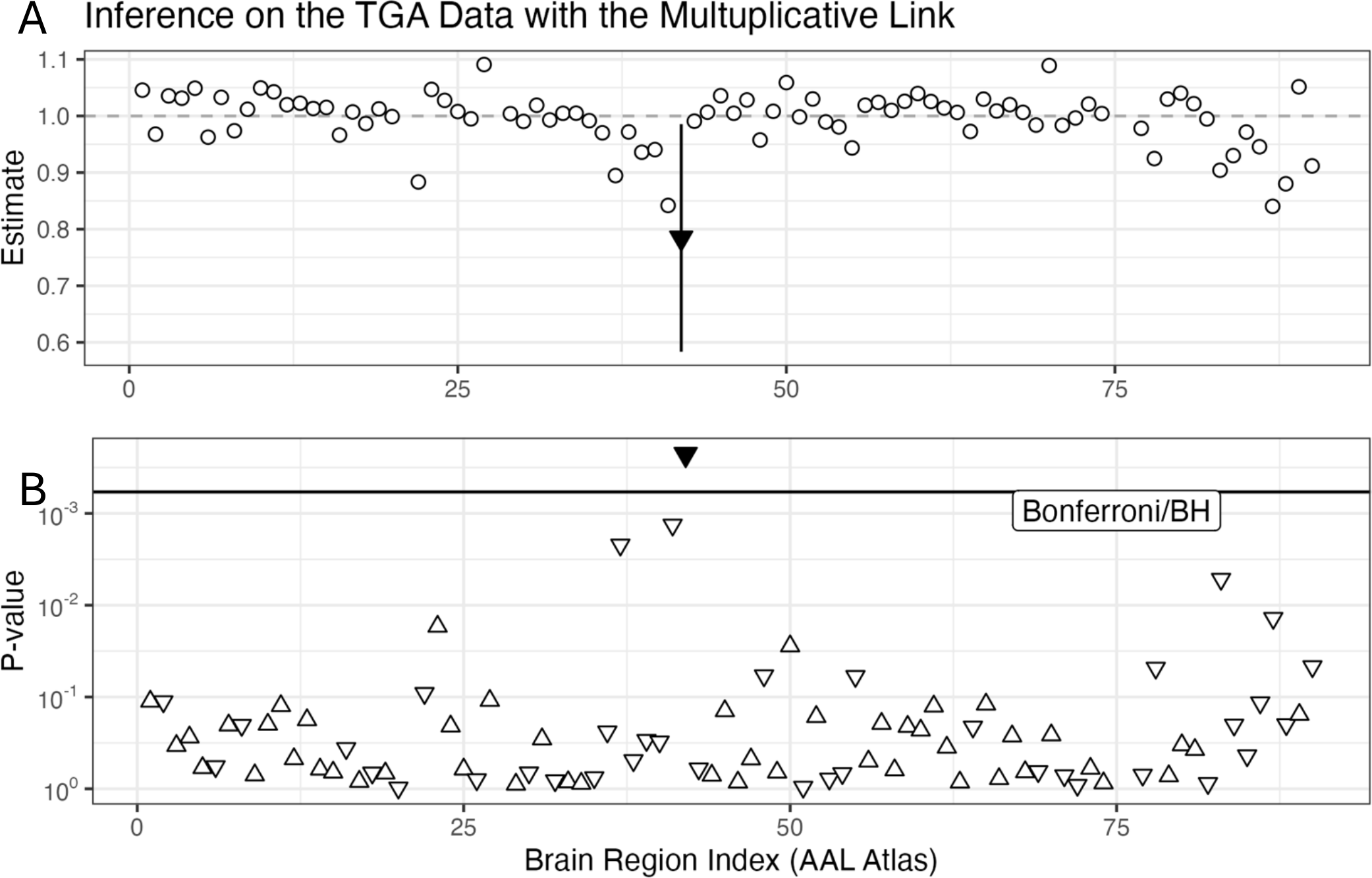}
        \caption{\small
        \textbf{Estimates of variable-specific perturbation effects in TGA patients.}
        A. Estimated variable effects $\hat{\alpha}_j$ with FCR-corrected confidence intervals
        Filled downward triangle indicates significance after multiplicity adjustments, and line indicates selection-corrected interval.
        B. Corresponding p-values (logarithmic scale) with  coinciding Bonferroni / BH threshold indicated.}
      \label{fig:tga-model-estimates}
    \end{center}
\end{figure}

\section{Discussion}

In this paper, we introduced a method for detecting and characterizing single-variable perturbations between populations of correlation matrices for RS-fMRI. The method represents differences between group mean correlation matrices through a structured  multiplicative link, reducing the dimensionality of the problem from quadratic to linear in the number of regions. Estimation is carried out using a covariance-aware generalized least-squares loss, and inference relies on robust variance estimation via a GEE framework applied to subject-level residuals. Together, this approach allows us to conduct region-level inference on correlational changes in RS-fMRI, while accounting for within-group variability and, to some extent,potential model misspecification. Beyond null hypothesis testing, selection-aware confidence intervals can be constructed for parameters with rejected nulls. Our simulation study demonstrates accurate parameter recovery, well-calibrated error control, and substantial gains in power relative to mass-univariate and sLED approaches. Application of our method to TGA data identified localized connectivity decays consistent with known memory-related regions.

The gains in statistical power arise from multiple elements in our framework. First, the structured multiplicative mean model pools information across correlation coefficients, substantially reducing the effective parameter space. Second, the generalized least-squares loss improves efficiency by incorporating dependence among correlation coefficients rather than treating them independently. Third, the residual-based GEE inference provides variance estimates that controls the variability reliably under misspecification, allows the method to retain the advantages of a parametric framework. Region-level interpretation and error control remain dependent on how well the structured mean model represents the signal. In the following paragraphs, we discuss how each component contributes to this balance between efficiency and stability.
The multiplicative mean model was motivated by an ANOVA-like decomposition of the between-group difference, where region-level perturbation effects interact to produce changes in the baseline correlations. Intuitively, if a given region loses synchronization, its correlations with other regions decay proportionally to their baseline magnitude. This structure contrasts with an additive formulation, $\rho_{ij} = \Theta_{ij} + \alpha_i + \alpha_j$, which may be appropriate when $\Theta$ is close to the identity but becomes unstable for dense correlation matrices and can lead to inadmissible (non–positive semidefinite) estimates.

The mean model is both the principal source of power and a potential limitation. By reducing the effective parameter space from order of $p^2$ to order of $p$, it pools signal across coefficients and enables region-level inference. At the same time, the one-dimensional perturbation assumption may be restrictive when multiple independent changes occur simultaneously. The current estimation approach can be viewed as projecting the group correlation matrices onto a low-dimensional manifold defined by $(\Theta, g(\Theta,\alpha))$, and conducting inference within this constrained space. Importantly, inference is based on residuals from the mean model, so their magnitude partially reflects deviations from the assumed mean structure. 

The second element in the method is the generalized least-squares loss. Correlation coefficients are strongly dependent, so the unweighted least squares loss is inefficient. As we discuss in the appendix, this dependence is a function of the mean correlation itself. We derive a matrix operator $C$ that represents the variability of estimating the correlation matrices, and use $C^{-1}$ to re-weight the errors. Importantly, consistency of the estimator does not require $C$ to equal the true covariance; however, efficiency improves as the approximation becomes more accurate.

The third element is the incorporation of the modeling error in both the mean model and GLS covariance when estimating the estimator variances. Estimation error is composed of (a) errors in estimating individual correlation matrices from noisy brain measurements, (b) the differences between individual correlation matrices within the group, and (c) the approximation error of the mean model. All components are reflected in the residuals of the individual observed correlation matrices from the model-based mean matrices. The GEE specifically accounts for the disparity between the empirical covariance of the observed residuals and the matrix $C$. Empirically, the GEE approach captures heterogeneity in standard errors across parameters. We observe a slight inflation in FWER rates, probably due to residual noise in the estimation of the variances, which we stabilize by adding a standard-deviation nugget. 

Investigating brain connectivity changes in disease states using RS-fMRI is common in the literature due to the relative ease in recording, which does not require patients' active co-operation. Nevertheless, these data suffer from low signal-to-noise, as well as patient heterogeneity and lack of temporal grounding. In addition, the patient data are usually compared with data from a separate group of healthy subjects, which poses the problem of between-group comparison. Still, this method is important as it can shed light on unique disorders. In such disorders, collecting task-data is difficult due to the rarity of the phenomenon, its relatively short appearance, or for lack of cooperation. 
TGA represents one such case, an elusive disorder whose pathophysiology is still a mystery, 
as the stark clinical symptoms are not necessarily accompanied by visible brain changes.
While a single lesion in the CA1 part of the hippocampus is now evident \citep{bartsch2006selective,Bartsch2010}, its appearance only 24-48 hours after symptoms' initiation, only in part of the patients and mostly unilaterally make their role in TGA's pathophysiology less conclusive. 
Functional MRI can provide important clues as to the origin of the disorder
 (e.g. \cite{levy1995functional,labar2002functional,peer2014reversible,kim2021aberrantly}). Our method was able to identify  major brain region involved in the early phase of TGA, with no need to insert prior assumptions about the neuro-anatomy.

A major challenge of the model is the computational cost of the procedure,
mainly the inversion of the correlation-covariance matrix for the weighted least squares penalty ($\invCb$).
In order to allow analyses of larger matrices,
this gap must be solved, either with the use of sparse matrices,
using the pseudo inverse or introduction of parallelization.
Additional technical improvements may include expanding the estimation to penalized-GEE, where specifically a ridge estimator can reduce the current large variances observed on non-null $\alpha$s due to strong negative correlations, and accounting for the noise in estimating the variances by an appropriate $t$ distribution. 

There are natural expansions and generalizations for the model.
First, throughout our paper we treated $\alpha$ as a vector in $\mathbb{R}^p$, thus limiting the number of hypothesis tests to the dimension of the correlation matrix in question.
However, $\alpha$ can be generalized to a matrix of size $p \times q$ where $q\geq1$, with a $\Theta \circ (\sum_{i=1}^q \alpha_{(i)}\alpha_{(i)}^t)+D_\alpha$ link function. The selection of $q$ then holds a trade-off between the flexibility of the model and the number of hypotheses to test.
Setting a large $q$ eases the severity of the model assumptions on the data while setting a low $q$ reduces the number of hypotheses. 

A second extension is to generalize the model so that it would adapt to more general experimental designs. Explicit correction of imbalances in individual covariates such as age and sex is important for two-group designs. Furthermore, the model should be adjusted for comparing subjects to themselves,
as many recent RS-fMRI studies feature intra-subject comparisons \citep{dinicola2020parallel, braga2020situating}.
Finally, when longer time-series are available, multiple correlation matrices can be estimated for each individual for different spectral dynamics. In these cases, we may be able to estimate decay in correlation at different temporal frequencies. 

In conclusion, we present here a method of discovering differences between groups of correlation matrices.
Our findings show that this new methodology can detect localized disturbances in the variables forming the correlation matrices.
We believe that this approach can be used to explore existing and new RS-fMRI datasets of different brain disorders, in which between-group comparison is of much importance, 
and may be useful in a wide array of additional fields, including neuroscience, economics, and engineering.

\section*{Software and data availability}
Our methods, as well as the TGA dataset, are available as an open-source R package \\ \emph{github.com/itamarfaran/corrpops}. We used version 1.0.0 for producing this manuscript. Our analysis code will be available upon publication as part of the supplementary materials.

\section*{Ethics statement}
This study was approved by the Institutional Review Board of Hadassah Medical Center (\#HMO44211) and conducted in accordance with the Declaration of Helsinki. All participants provided written informed consent prior to enrollment. All MRI data were fully anonymized prior to analysis.

\section*{Usage of Large Language Models}
The original manuscript used large language models only for the final proof corrections. In the revision, ChatGPT (www.chatgpt.com, versions 5.5, 5.6) was used for ideation, generation of experiment code, and language editing. The authors are responsible for all the content of the manuscript.

\section*{Acknowledgements}
We would like to thank two anonymous reviewers whose thoughtful comments greatly improved the method and the manuscript. We thank Yuli Slavutsky, Lee Carlin, Omer Ronen, David Zucker, and Yoav Benjamini for their comments on early versions of this manuscript, and Lihua Lei for valuable discussions.

The work was supported by the Center for Interdisciplinary Data Science Research (CIDR).  
IF and YB were supported in part by the National Institutes of Health (NIH) R01GM083084, and the German Israeli Foundation. MP was supported by an Eshkol fellowship, an Eva, Luis and Sergio Lamas Fund award, a Fulbright fellowship, and the Zuckerman STEM Leadership Program fellowship.

\section*{Declaration of Competing Interests}

No competing interest to declare.

\printbibliography

\appendix

\section{Proofs from Section 2.3}
\subsection{Proofs of Conditions for Positive Semi-Definiteness}
\label{subsec:PSD-Proof}
\begin{proposition}
    Define $\lambda_{p}\geq0$ to be the minimal eigenvalue of $\Theta$, and $\alpha_{M},\alpha_{m}$ to be the maximal and minimal values of $\left\{ \alpha_{1},...,\alpha_{p}\right\}$, respectively.
    For the multiplicative link function $g$:
    \begin{enumerate}
        \item For arbitrary $\alpha$,
        if $\lambda_{p} \cdot \alpha_m^2 \geq \alpha_{M}^{2}-1 $, then $g(\Theta, \alpha)$ is positive semi-definite.
        \item Specifically, if $\alpha_M \leq 1$,then $g\left(\Theta, \alpha\right)$ is positive semi-definite.
    \end{enumerate}
\end{proposition}
\begin{proof}
Let $\Theta\in\mathcal{C}_{p}$ a correlation matrix,
and $\alpha\in\mathbb{R}^{p}$. Consider the multiplicative link function $g$ in its matrix form:
\[
    g\left(\Theta,\alpha\right)
    =\Theta\circ\alpha\alpha^{t}+I-D_{\alpha}^{2}
\]
with $D_{x}: \mathbb{R}^p \mapsto \mathbb{M}_p\left(\mathbb{R}\right)$ defined as 
\[
    \left[D_{x}\right]_{ij}=\begin{cases}
    x_{i} & i=j,\\
    0 & otherwise.
    \end{cases}.
\]
We wish to prove that $x^{t}g\left(\Theta,\alpha\right)x\geq0\,\forall\,x\in\mathbb{R}^{p}$.
Write the spectral decomposition of $\Theta$ as $\Theta=U\Lambda U^{t}$, with $\lambda_{1}\geq\lambda_{2}\geq...\geq\lambda_{p}\geq0$.
Also, recall that:
\begin{align}
    x^{t}\left(A\circ B\right)y
    &= \text{tr}\left(D_{x}AD_{y}B^{t}\right)
    \label{eq:hadamard-trace}
    \\
    \max_{j}a_{j}\cdot||x||^{2}
    &\geq x^{t}D_{a}x\geq\min_{j}a_{j}\cdot||x||^{2} \, \forall \, a
    \label{eq:norm-bound}
\end{align}
where \ref{eq:norm-bound} is a result of
\[
    x^{t}D_{a}x=\sum_{i}a_{i}x_{i}^{2}
    \geq\sum_{i}\min_{j}a_{j}\cdot x_{i}^{2}=\min_{j}a_{j}\cdot||x||^{2}.
\]
Similarly, $x^{t}D_{a}x\leq\max_{j}a_{j}\cdot||x||^{2}$.
We wish to find the condition on which $g\left(\Theta,\alpha\right)$
is positive semi-definite, $x^{t}g\left(\Theta,\alpha\right)x\geq0$:
\begin{align*}
    x^{t}g\left(\Theta,\alpha\right)x
    &= x^{t}\left(\Theta\circ\alpha\alpha^{t}+I-D_{\alpha}^{2}\right)x \\
    &= x^{t}\left(\Theta\circ\alpha\alpha^{t}\right)x+x^{t}x-x^{t}D_{\alpha}^{2}x.
\end{align*}
Now, using \ref{eq:hadamard-trace}, 
\begin{align*}
    x^{t}\left(\Theta\circ\alpha\alpha^{t}\right)x 
    &= \text{tr}\left(D_{x}\Theta D_{x}\alpha\alpha^{t}\right) \\
    &= \alpha^{t}D_{x}\Theta D_{x}\alpha \\
    &= \alpha^{t}D_{x}U\Lambda U^{t}D_{x}\alpha
\end{align*}
And using \ref{eq:norm-bound},
\begin{align*}
    \alpha^{t}D_{x}U\Lambda U^{t}D_{x}\alpha
    &\geq \lambda_{p}\cdot\alpha^{t}D_{x}UU^{t}D_{x}\alpha \\
    &= \lambda_{p}\cdot\alpha^{t}D_{x}D_{x}\alpha \\
    &= \lambda_{p}x^{t}D_{\alpha}^{2}x.
\end{align*}
Putting all the above together yields:
\begin{align*}
    x^{t}\left(\Theta\circ\alpha\alpha^{t}\right)x+x^{t}x-x^{t}D_{\alpha}^{2}x
    &\geq \lambda_{p}x^{t}D_{\alpha}^{2}x+x^{t}x-x^{t}D_{\alpha}^{2}x \\
    &\geq \lambda_{p}\alpha_{m}^{2}||x||^{2}+||x||^{2}-\alpha_{M}^{2}||x||^{2} \\
    &\propto \lambda_{p}\alpha_{m}^{2}+1-\alpha_{M}^{2}
\end{align*}
So for $x^{t}g\left(\Theta,\alpha\right)x$ to be non-negative, a sufficient condition is:
\[
    \lambda_{p} \cdot \alpha_m^2 \geq \alpha_{M}^{2}-1 \Rightarrow
    x^{t}g\left(\Theta,\alpha\right)x\geq0.
\]
\end{proof}

\textbf{Notes:} \\
\begin{itemize}
\item If $0<\left|\alpha_{j}\right|\le1$ for all $j$,
$g\left(\Theta,\alpha\right)$ is \emph{always} positive semidefinite
as $\frac{\alpha_{M}^{2}-1}{\alpha_{m}^{2}}\leq0$ and $\lambda_p\geq0$.
\item If $\max_{j}\alpha_{j}>1$, $\Theta$ must be positive definite ($\lambda_{p}>0$).
\end{itemize}
\subsection{Proof of Identifiability}
\label{subsec:Model-Identifiability-Proof}
We provide here more general conditions for identifiability of the multiplicative model, which accommodate cases where $\alpha_j = 0$ for some of the coordinates of $\alpha$. 
The main difficulty for identifying the multiplicative model are zero values in either $\Theta$ or $\alpha$. For example, when the $\Theta$ is the identity matrix,  $\alpha$ doesn't affect the model at all. Conversely, if $\alpha_1,...,\alpha_{p-1}=0$, the value of $\alpha_p$ does not affect the model, as all correlation coefficients in $\Theta \circ \alpha \alpha^{t}$ would regardless decay to 0.  

The main idea is that each non-zero entry in $\Theta$ gives some information on the two relevant values of $\alpha$. We therefore introduce matrices $A(M)$ that, for a set of indices $M$, code all index pairs for which $\Theta$ is non-zero. 

Introducing notation, denote the elementary vector by $e^{\left(j\right)}\in\mathbb{R}^p$
meaning that $ e^{\left(j\right)}_{i}=  1$ if $ i=j$ and $e^{\left(j\right)}_{i}=0$ otherwise, 
and write $e^{\left(j,k\right)} = e^{\left(j\right)} + e^{\left(k\right)}$.
Then form matrix $A$ of size $m\times p$ ($m = \binom{p}{2}$), where the $\ell$'th row is defined as:
\begin{equation} A_{\ell\cdot} = \boldsymbol{1}\left\{\Theta_{j_{\ell} k{_\ell}} \neq 0 \right\} \cdot e^{\left(j_{\ell},k{_\ell}\right)}, \quad  \ell = 1,...,m.
\label{eq:DefA}
\end{equation}
Index $\ell$ iterates over the ordered index pairs $\left(j_\ell,k_\ell\right) = \left(1,2\right), \left(1,3\right),...,\left(p-1,p\right)$. 
For a subset of indices $M \subseteq \{1,...,p\}$, we define $A(M)$ to be a filtered version of $A$, 
keeping only those rows where both indices are in $M$, and columns which correspond to $M$. ($A(M)$ would therefore have $\binom{|M|}{2}$ rows and $|M|$ columns). 

\begin{proposition}
Consider the model \eqref{eq:fullmodel} using the multiplicative link \\ parametrized by $(\Theta, \alpha)$, and further assume the coordinates of $\alpha$ are non-negative.   Let $M_+ = \{j: \alpha_j > 0\}$ be the set of indices where $\alpha$ is non-zero, and $M_0= \{j: \alpha_j = 0\}$ the complementary set. Let $p_+$ be the size of $M_+$. \\
Then the following conditions are sufficient for identifiability: 
\begin{enumerate} 
\item For every $j \in M_0$, exists $k \in M_+$ so that $\Theta_{jk} \neq 0$.
\item The rank of the matrix defined on the non-zero $\alpha$ coordinates $A(M_+)$ is $p_+$. 
\end{enumerate}
\end{proposition}

\begin{proof}
Denote $T\left(Y\right)=\left\{ \bar{R}_{A},\bar{R}_{B}\right\}$. 
Our goal is to show that the map $(\Theta, \alpha) \mapsto E\left[T\left(Y\right)\right]$ is one-to-one. 

Under the multiplicative link, $E\left[T\left(Y\right)\right]=\left\{ \Theta,\Theta\circ\alpha\alpha^{t}\right\}$. Therefore $E[T_1(Y)]=\Theta$ is trivially one-to-one with respect to $\Theta$. 
We are left to prove that for $\alpha,\beta\in\mathbb{R}^{p}$,
if $\Theta\circ\alpha\alpha^{t}=\Theta\circ\beta\beta^{t}$ then $\alpha=\beta$.

First, we show that the zero-set for vectors $\alpha$ and $\beta$ are identical.
Consider a non-zero $\alpha$ coordinate ($j \in M_+$).
According to condition (2), $A(M_+)$ should be of rank $p_+$, but only columns corresponding to indices in $M_+$ are non-zero. Hence, there should be at least one non-zero entry for each column of $A(M_+)$. In particular, for $j 
\in M_+$ there exists $k\in M_+$ such that $\Theta_{jk} \neq 0$. Therefore $\Theta_{jk}\alpha_j\alpha_k \neq 0$. Then $\Theta_{jk}\beta_j\beta_k \neq 0$ and in particular $\beta_j \neq 0$. Next consider a coordinate $\alpha_j \in M_0$. From condition (1), there exists coordinate $k \in M_+$ so that $\Theta_{jk}\neq 0$. In particular, $0 = \alpha_j \alpha_k \Theta_{jk} = \beta_j \beta_k \Theta_{jk}$. $\Theta_{jk}, \beta_k \neq 0$, so $\beta_j = 0$. 

We are left to show that for $j \in M_+$, $\alpha_j = \beta_j$. We have $p_+$ such parameters,
and without loss of generality assume $M_+ = \{1,...,p_+\}$.
For any $j,k \in M_+$ for which $\Theta_{jk} \neq 0$, we can form the equation
\[ |\Theta_{jk}|  \cdot \alpha_j \cdot \alpha_k = |\Theta_{jk}|  \cdot \beta_j \cdot \beta_k \]
Taking the natural logarithm, 
\[ \ln\left|\Theta_{jk}\right|  + \ln\alpha_j  + \ln\alpha_k = \ln\left|\Theta_{jk}\right| + \ln\beta_j +
 \ln\beta_k .\]

Denoting $c_j = \ln\alpha_j - \ln\beta_j$, we get for each non-zero index in $\Theta$ an equation $c_j + c_k= 0$. Gathering these constraints, we get a linear equation set in $(c_1,...,c_{p_+})$ with a design matrix equivalent to $A(M_+)$ \eqref{eq:DefA}:
\[
    \underset{\tilde{m}\times p_+}{
        A (M_+)
    }\cdot
    \underset{{p_+}\times1}{
        \left[
            \begin{array}{c}
                c_{1}\\
                \vdots\\
                c_{p_+}
            \end{array}
        \right]
    }=
    \underset{\tilde{m}\times1}{
        \left[
            \begin{array}{c}
                \vdots\\
                0\\
                \vdots
            \end{array}
        \right].
    }
\]
Here $\tilde{m} = \binom{p_+}{2}$. Given that $A(M_+)$ is full rank, the only solution is $(c_1,...,c_{p_+}) = \mathbf{0}$, meaning that $\ln\alpha_j = \ln\beta_j$ for every $j \in M_+$. We conclude $\alpha_j = \beta_j$ from the monotonicity of the $\ln$ function.
\end{proof}

\section{Moments of correlation matrices of Gaussian vectors}
\label{appendix:correlMat}
Here we review the moments of correlation matrices estimated from correlated Gaussian vectors (for general results see 
\cite{neudecker1990asymptotic} and \cite{schott2005matrix};
for the effect of temporal correlation on inference see \cite{afyouni2019effective} and \cite{efron2008row, efron2012large}).
We use these results in estimating our model.

In the following, we denote the Matrix-Gaussian distribution with $Y\sim\mathcal{MN}_{T \times p}\left(\boldsymbol{M},\Delta,\Sigma\right)$, which means that $Y$ is jointly Gaussian where each row is distributed as a Gaussian $Y_{t \cdot}\sim \mathcal{N}_p(\boldsymbol{M}_{t\cdot},\Delta_{tt} \cdot \Sigma)$, and each column is distributed as a Gaussian $Y_{\cdot j}  \sim \mathcal{N}_T(\boldsymbol{M}_{\cdot j},\Delta \cdot \Sigma_{jj})$.

Let $X$ be a random data matrix of size $T \times p$,
and assume that $X\sim\mathcal{MN}_{T\times p}\left(\boldsymbol{0},\Delta,\Sigma\right)$.
When the rows are \emph{i.i.d} ($\Delta=Id_{T}$), $W=cov\left(X\right)$ is the maximum likelihood estimator for $\Sigma$,
and $T\cdot W$ follows a $Wishart_{p}\left(\Sigma,T\right)$ distribution.
The first two moments of $W$ are then given by
$E\left[W\right]  =\Sigma,$ and $T\cdot Cov\left(W_{ij},W_{kl}\right)  =\Sigma_{ij,kl}^{\left(2\right)}=\Sigma_{ik}\Sigma_{jl}+\Sigma_{il}\Sigma_{jk}$,
as shown for example in \cite{efron2012large}.
However, the elements of the empirical correlation matrix
have distinctly smaller variances due to the scaling.
The following lemma describes the limiting covariance of values in the empirical correlation matrix: 

\begin{lemma}
    \label{lemma:cor-cov}
    The asymptotic moments of the empirical correlation matrix for i.i.d samples.
    Let $R=corr(X)$ for $X\sim\mathcal{MN}_{T\times p}\left(\boldsymbol{0},I_{T},\Sigma\right)$.
    Denote by $\boldsymbol{\rho} = scale(\Sigma)$, or
    $\boldsymbol{\rho}_{ij}=\Sigma_{ij}\slash\sqrt{\Sigma_{ii}\Sigma_{jj}}$.
    Then $\underset{T\rightarrow\infty}{\lim}E\left[R\right]=\boldsymbol{\rho},$ and $\underset{T\rightarrow\infty}{\lim}T\cdot Cov\left(R_{ij},R_{kl}\right)=C_{ij,kl}\left(\boldsymbol{\rho}\right)\:$, where
    \begin{equation}
        C_{ij,kl}\left(\boldsymbol{\rho}\right)\coloneqq\frac{\rho_{ij}\rho_{kl}}{2}\left(\rho_{ik}^{2}+\rho_{il}^{2}+\rho_{jk}^{2}+\rho_{jl}^{2}\right)-\rho_{ij}\left(\rho_{ik}\rho_{il}+\rho_{jk}\rho_{jl}\right)
        \label{eq:cor-cov}
    \end{equation}
    \[-\rho_{kl}\left(\rho_{ik}\cdot\rho_{jk}+\rho_{il}\cdot\rho_{jl}\right)+\left(\rho_{ik}\rho_{jl}+\rho_{il}\rho_{jk}\right).\]
\end{lemma}

The lemma was introduced by \cite{pearson1898vii} and is listed by \cite{nel1985matrix} and \cite{schott2005matrix},
and is derived using a first-order approximation of the $corr\left(\cdot\right)$ function on the empirical covariance matrix.
The convergence is of the moments of $R$, in that higher order terms vanish as sample size increases. The element-wise operations in \eqref{eq:cor-cov} can be written as a matrix operator $C:\mathcal{C}_p \to M_m(R)$: 
\begin{equation}
C(\vec{\rho}) \stackrel{\cdot}{=} T \cdot Cov(\vec{R}).
\end{equation}
Later, we will plug in the observed correlation matrices $R_i$ into the operator $C$.

\subsubsection*{Accounting for Temporal Correlation}
In fMRI time-series data (or other data characterized by temporal correlation), the correlation between adjacent rows can be large,
and hence less information is added from each row compared to the \emph{i.i.d} case.
While $W$ remains an unbiased estimator of $\Sigma$, the covariance of $W$ and $R$ will be larger than the expression in \eqref{eq:cor-cov}.
The following correction, due to \cite{efron2012large},
identifies an effective degrees of freedom $T_{eff}$, which can be estimated directly from the scaled temporal covariance matrix $\Delta$.
We cite the result with no proof. 

\begin{lemma}
    \label{lemma:cor-cov-dep}
    The asymptotic moments of empirical correlation matrices for correlated samples.
    Let $X\sim\mathcal{MN}_{T\times p}\left(\boldsymbol{0},\Delta,\Sigma\right)$
    for a general positive-definite $\Delta$,
    and let $W = cov(X)$ and $R = corr(X)$. 
    For $P_\Delta = scale(\Delta)$, define $\psi(P_\Delta)$ to be the mean-square correlation:
    $\psi^2=\binom{T}{2}^{-1}\sum_{t<s}(P_\Delta)_{ts}$. 
    Then the covariance of the correlation matrix shrinks like the effective degrees of freedom $T_{eff}(T,\psi)=T\slash\left(1+\left(T-1\right)\cdot\psi^{2}\right)$, 
    namely 
    \begin{align*}
    &T_{eff}(T,\psi)\cdot Cov\left(W_{ij},W_{kl}\right)=\Sigma_{ij,kl}^{\left(2\right)},\\
    \underset{T_{eff}\rightarrow\infty}{\lim}&T_{eff}(T,\psi)\cdot Cov\left(R_{ij},R_{kl}\right)=C_{ij,kl}\left(\boldsymbol{\rho}\right).
    \end{align*}
\end{lemma}
To summarize the above lemmas, the first two moments of the vectorized empirical correlation matrix $\hat{R}$ can be described
\begin{equation}
    E\left[\vec{R}\right] \underset{T_{eff}\to \infty}{\longrightarrow} \vec{\boldsymbol{\rho}},\qquad Cov\left[\sqrt{T_{eff}}\cdot\left(\vec{R}-\vec{\boldsymbol{\rho}}\right)\right] \underset{T_{eff}\to \infty}{\longrightarrow} C(\boldsymbol{\rho}). 
    \label{eq:corr-clt}
\end{equation}

\section{Regularization of the Inverse Covariance}
\label{appendix:regularization}

Here we discuss the rational and tradeoffs governing the regularization of the working covariance matrix. The regularized working inverse covariance is $C^*_B = \left((1-\lambda)C_B+\lambda \bar d \cdot \mathrm{Id}_m \right)^{-1}$,  where $\bar d$ is the average diagonal element of $C_B$ and $\lambda$ is the regularization constant.  Thus, $\lambda=0$ corresponds to full inverse-covariance weighting, whereas $\lambda=1$ makes $C_B^*$ proportional to the identity and therefore reduces $S^B$ in Equation \ref{eq:alpha_k} to an ordinary least-squares criterion. 

Decomposing the loss according to the eigenvectors ($q_\ell$) and eigenvalues $\mu_\ell$ of $C_B$  reveals the following form
\[\sum_{i\in B}\sum_{\ell=1}^{m}
\frac{
\left[
q_\ell^t
\left(\vec{R}i-\vec{g}(\Theta,\alpha)\right)
\right]^2
}{
(1-\lambda)\mu_\ell+\lambda\bar d
}.
\]
Without regularization ($\lambda = 0)$, discrepancies in directions associated with small eigenvalues of $C_B$ receive very large weights, whereas discrepancies in directions with greater sampling variability receive relatively little weight. This weighting is efficient when the mean model is correctly specified. Under misspecification, however, it can make the fitted parameters overly sensitive to relatively minor discrepancies in low-variance directions. In particular, broad increases or decreases across many correlations, which are the patterns that the parameters $\alpha$ are intended to represent, may be concentrated in directions associated with larger eigenvalues of $C_B$. Full inverse-covariance weighting may then favor values of $\alpha$ that reduce small residual components while providing a poor reconstruction of the dominant observed differences under the ordinary norm.

Importantly, under correct model specification, regularization does not alter the population target nor the validity of the inference procedure. If $E\left(\vec{R}i\mid i\in B\right)=\vec{g}(\Theta_0,\alpha_0)$ then the population mean residual vanishes at $(\Theta_0,\alpha_0)$. Furthermore, GEE inference is adaptive to the covariance matrix used for the optimization; however, the results variance estimates are larger as regularization increases. Using the regularized covariance therefore provides a compromise between robustness and efficiency: the fitted values of $\alpha$ continue to reflect the dominant changes in the observed correlations, while covariance weighting retains much of the reduction in sampling variance. In our simulations, complete regularization produced larger variances for the estimated parameters, whereas intermediate regularization maintained relatively accurate estimation and was more robust to misspecification. Empirically, the results were relatively insensitive to the degree of regularization, with only minor differences for $\lambda$ between $0.05$ and $0.5$.

\section{Fast computation of GEE variance for $\alpha$}
\label{appendix:block_gee}

The Schur-complement formulation provides an efficient way to calculate the covariance of a subset of the parameters, in our case $\hat{\alpha}$, while accounting for the uncertainty in the nuisance parameters $\Theta$. 
Due to properties of the multiplicative link, this computation only requires inverting a $p\times p$ matrix, 
rather than the much larger $m+p\times m+p$ matrix. 

The Jacobian of the joint estimating equations for $\eta = \left(\vec{\Theta}, \alpha\right)$ can be written in block form
$$I = \left(
\begin{array}{cc}
I_{\Theta\Theta}& I_{\Theta\alpha}\\
I_{\alpha\Theta} & I_{\alpha\alpha} 
\end{array}
\right).$$

The Schur complement of 
$I_{\Theta\Theta}$ is
$S_{\alpha} = I_{\alpha\alpha} -
I_{\alpha\Theta}
I_{\Theta\Theta}^{-1}
I_{\Theta\alpha}.$ This is the effective bread matrix for \(\alpha\) after accounting for the joint estimation of the nuisance parameter \(\Theta\).

For the multiplicative link, 
recall that the estimating equation for $\Theta$ is $U_\Theta(\eta) = \sum_{i\in \mathcal{A}} (\vec{R}_i - \vec{\Theta}) + \sum_{i\in \mathcal{B}} (g^{-1}(\vec{R}_i ,\alpha) - \Theta)$, 
with a linear dependence on $\Theta$. Therefore 
$I_{\Theta\Theta} =\frac{\partial U_\Theta}{\partial \Theta^t} = -(n_A + n_B) Id_m$ 
for $Id_m$ the identity matrix, 
and its inverse is $I_{\Theta\Theta}^{-1} = -\frac{1}{n_A+n_B}Id_m$.

Write the contribution of subject \(i\) to the joint estimating function as
$u_i(\eta)=
(u_{\Theta i}(\eta)^t , u_{\alpha i}(\eta)^t )^t$.  
Evaluating all quantities at \(\widehat{\eta}\), define the adjusted contribution for \(\alpha\) by $
\widetilde{u}_{\alpha i}
=
u_{\alpha i}
-
I_{\alpha\vec{\Theta}}
I_{\vec{\Theta}\vec{\Theta}}^{-1}
u_{\vec{\Theta} i}. $
The corresponding empirical meat matrix is
\[
\widetilde{\Upsilon}_{\alpha}
=
\frac{n_A}{n_A-1}
\sum_{i\in A}
\widetilde{u}_{\alpha i}
\widetilde{u}_{\alpha i}^{\mathsf{T}}
+
\frac{n_B}{n_B-1}
\sum_{i\in B}
\widetilde{u}_{\alpha i}
\widetilde{u}_{\alpha i}^{\mathsf{T}}.
\]
The sandwich covariance estimator for \(\widehat{\alpha}\) can therefore be
written as
\[
\widehat{V}_{\mathrm{gee}}(\widehat{\alpha})
=
S_{\alpha}^{-1}
\widetilde{\Upsilon}_{\alpha}
S_{\alpha}^{-t}.
\]

This expression is algebraically equivalent to taking the lower-right
\(p\times p\) block of the full joint sandwich estimator
\[
\widehat{V}_{\mathrm{gee}}(\widehat{\eta})
=
I^{-1}\Upsilon I^{-t},
\]
but avoids explicitly constructing and inverting the full matrix. The
adjustment to \(u_{\alpha i}\) also shows directly how uncertainty in estimating
\(\vec{\Theta}\) contributes to the covariance of \(\widehat{\alpha}\).

\section{Misspecified mean experiments}
\label{appendix:misspecified}
We compared the proposed GEE procedure with edgewise two-sample t-tests and sLED under two forms of model misspecification. In each simulation, $p=32$ variables were observed in two equally sized groups ($n_A = n_B = 50$). A random population correlation matrix was generated for each replicate, and selected correlations in the second group were moved toward zero by $\epsilon$, while preserving their signs and maintaining $|\rho|\geq 0.03$. 

We considered two alternatives: (a) a sub network comprising of 10 variables ($\approx 30\%$) for which all within-network correlations decayed towards zero in the case population; and (b) sparse alternatives in which 2, 6, or 10 individual correlation coefficients decayed towards zero. Effect sizes increased from a null change ($\epsilon=0$) to $\epsilon=0.02,0.04,$ and 0.06. The positive set included all variables that had correlation changes. The same baseline matrix and random sampling stream were used across effect sizes within each replicate. Candidate networks or edge-sets were retained only when the resulting population correlation matrix remained positive definite. Multiplicity corrections followed the BH procedure, and rejections corresponded to $p_{adj}< 0.05$. Performance was summarized by the probability of detecting any signal (Power), the percentage of changed edgest that were rejected (true discovery proportion, TDP), and the percentage of unchanged edges that were rejected (false discovery rate, FDR).

At the top row of Table 2 shows data sampled under the null, with no changed edges. Here \emph{Power} corresponds to false rejections, and we see that all methods are calibrated. For the Network analysis, we see that our method is more powerful to detect a change from the null, and also has high rejection rates of variables in the network. On the other hand, it shows increased false rejection of variables not in the network; this probably occurs because the multiplicative model does not describe accurately the difference between the groups. In the correlation models, the method is still relatively powerful compared to the T statistics and SLED, though when there are more correlation coefficients to detect the T statistist is more suited for detecting the changes correctly. 

\begin{table}[ht]
\centering
\begin{tabular}{llrrrrrrrr}
  \toprule
Scenario & Affected set & $\varepsilon$ & \multicolumn{3}{c}{GEE} & \multicolumn{3}{c}{Edgewise T tests} & \multicolumn{1}{c}{sLED} \\
\cmidrule(lr){4-6} \cmidrule(lr){7-9} \cmidrule(lr){10-10}\\
 & & & Power & TDP & FDR & Power & TDP & FDR & Power \\ 
\midrule
 Null & \- & 0 & 0.040 & - & - & 0.000 & - & - & 0.020 \\ 
  Network & 10 vars & 0.02 & 0.470 & 0.070 & 0.045 & 0.020 & 0.001 & 0.040 & 0.020 \\ 
  Network & 10 vars & 0.04 & 0.920 & 0.349 & 0.162 & 0.200 & 0.006 & 0.035 & 0.120  \\ 
  Network & 10 vars & 0.06 & 0.980 & 0.569 & 0.211 & 0.660 & 0.055 & 0.042 & 0.450  \\ 
  Coefficients & 2 coefs & 0.02 & 0.000 & 0.000 & 0.060 & 0.000 & 0.000 & 0.040 & 0.020 \\ 
  Coefficients & 2 coefs & 0.04 & 0.030 & 0.007 & 0.070 & 0.000 & 0.000 & 0.040 & 0.010  \\ 
  Coefficients & 2 coefs & 0.06 & 0.070 & 0.018 & 0.070 & 0.050 & 0.030 & 0.050 & 0.020 \\ 
  Coefficients & 6 coefs & 0.02 & 0.050 & 0.005 & 0.045 & 0.000 & 0.000 & 0.040 & 0.010\\ 
  Coefficients & 6 coefs & 0.04 & 0.150 & 0.018 & 0.058 & 0.010 & 0.002 & 0.040 & 0.020  \\ 
  Coefficients & 6 coefs & 0.06 & 0.240 & 0.030 & 0.090 & 0.200 & 0.043 & 0.048 & 0.060  \\ 
  Coefficients & 10 coefs & 0.02 & 0.060 & 0.004 & 0.040 & 0.010 & 0.001 & 0.035 & 0.010 \\ 
  Coefficients & 10 coefs & 0.04 & 0.220 & 0.020 & 0.065 & 0.060 & 0.006 & 0.040 & 0.050 \\ 
  Coefficients & 10 coefs & 0.06 & 0.350 & 0.038 & 0.083 & 0.440 & 0.076 & 0.032 & 0.150  \\ 
   \bottomrule
\end{tabular}
\caption{Power and error rates across simulation scenarios.} 
\label{tab:power-comparison}
\end{table} 

\section{Resampling analysis for TGA data}
\label{appendix:bootstrap_tga}

We conducted the bootstrap analysis on $B=300$ bootstrap replicates resampled from the original data. Each dataset was composed of $n_A = 23$ control and $n_B=14$ patient correlation matrices ($p=86$) that were sampled with replacement from the respective groups. Parameters $(\Theta^b, \alpha^b$) were estimated for each dataset. 
We computed the empirical standard deviation for each region from the estimated parameters 
\[\hat{\sigma}_j  = \sqrt{\frac{1}{B-1}\sum_b^B(\alpha
_j^b - \bar{\alpha}_j)^2},\]
with $\bar{\alpha}_j$ the average over replicates. In the analysis, 
the bootstrap estimate is compared to the GEE-based standard deviation that was used for inference $s_j +s_0 = \sqrt{\hat{V}_{gee,jj}}+s_0$. The results are shown in Figure \ref{figure:bootstrap}. We report the Pearson correlation across the 86 regions, and the average relative difference: 
\[\delta = \frac{1}{p}\sum_j\frac{|\hat{\sigma}_j - \tilde{s}_j|}{\hat{\sigma}_j}.\]

We used a permutation analysis to tabulate the expected distributions of estimates $\alpha$ under the null hypothesis, when the data has similar properties to the observed. 
To generate null data, we rotated the patient data so that it would be centered around the observed control data. For $W_A = \bar{\Lambda}_A, W_b = \bar{\Lambda}_B$ the observed mean correlation matrices, we found the roots and inverse roots using eigenvalue decompositions. 
Then, we rotated the patient matrices and rescaled them as follows:
\[R^w_i = scale(W^{1/2}_A (W^{-1/2}_B R_i  W^{-1/2}_B ) W^{1/2}_A),\] 
whereas the control matrices where left unchanged. This transformation results in well determined correlation matrices for both groups, and very similar average correlation matrices for the two groups.

We next generated $B=300$ datasets by randomly assigning (without replacement) the $n_A+n_B=37$ correlation matrices into the two groups. For each new null dataset, we estimated the model and recorded the resulting $\alpha^b$ vector. We tested the hypothesis $H_{1j}: \alpha_j\neq 1$ by measuring the magnitude of change from 1, $t(a)=|1-a|$.
P-values where adjusted for FWER by comparing them to the maximum statistic of each permutation. Explicitly, 
\[p_{j,adj} = \frac{1}{B} \sum_{b=1}^B 1\left(t\left(\hat{\alpha}_j\right) \leq max_j\left\{t\left(\alpha_j^b\right)\right\}\right).\]

\section{Sampling procedures for simulations}
\label{section:sampling_algo}

The simulation data are generated in three stages: (1) generation of population-level parameters, (2) construction of covariance matrices with optional temporal dependence, and (3) generation of sample correlation matrices.

We first generate a population correlation matrix and a vector of scaling parameters. To obtain a random positive-definite matrix, we draw a $2p\times p$ matrix of standard normal variables, $
x \sim \mathcal{MN}_{2p \times p}(0, I_{2p}, I_p),$
and form the Gram matrix $x^\top x$. This matrix is converted to a correlation matrix by standardizing with respect to its diagonal elements $\Theta = scale(x^tx)$.

Next, we construct the perturbation effects vector $\alpha \in \mathbb{R}^p$ as a mixture of $p\cdot \alpha_{prop}$ variables sampled uniformly between $\alpha_{min}$ and $\alpha_{max}$. The remaining coefficients are set to $1$. This creates a mixture of variables with heterogeneous and baseline scaling factors.

To investigate the role of temporal autocorrelation, the data were sampled according to temporal covariance matrices under two conditions: (a) identity covariance or (b) a multivariate $ARMA\left(1,1\right)$ process with both AR and MA parameters equal to .5. In practice we estimate the appropriate covariance matrix $\Delta$ from the process, and use it for sampling the multivariate data. 

\subsection{Generation of Sample Correlation Matrices}
To generate the sample correlation matrices, we first set the population mean correlation matrix for the two groups: $\Lambda^A = \Theta$ and $\Lambda^B = g(\Theta,\alpha)$. 

Then, for each group, we generate the individual matrices. For the temporal independence case, 
we sample directly from a Wishart before rescaling: 
\[
W_i \sim \text{Wishart}_p(\Lambda, T), \quad R_i = scale(W_i).
\]
For the autocorrelated case, we sample ARMA processes and estimate the resulting matrix empirically: 
\[  Y_i \sim  MN(\mathbf{0}, \Delta,\Lambda) \quad R_i = scale(cov(Y_i)).
\]

\subsection{Simulation with additional subject-level correlation variability}
\label{appendix:subjectvar}

We evaluated whether the GEE variance estimator reflects increased between-subject variability in correlation structure. Each simulated dataset contained $p=40$ variables, $n_A=n_B=30$ subjects per group, and subject-level sample correlation matrices based on $T=120$ observations. We generated three fixed random correlation designs; Figure \ref{fig:sub_variability} displays average results over the three designs.

For each design, we generated a positive-definite baseline correlation matrix $\Theta$ from a five-factor covariance model, rescaled it to a correlation matrix and mildly shrunk it towards the identity. We considered both a null setting, with $\alpha_j=1$ for all $j$, and a non-null setting, in which the node-specific effects $\alpha_j$ were drawn once from $\mathrm{Unif}(0.78,0.88)$ and then held fixed. The population correlation matrix in the diagnosed group was obtained from the multiplicative model,
\[
R^{(B)}_{jk} = \Theta_{jk}\alpha_j\alpha_k, \qquad j\ne k,
\]
whereas the control-group population matrix was $\Theta$. Thus, the population contrast was fixed across all variability levels.

We then independently perturbed the population correlation matrix of each subject in both groups according to one of two mechanisms. Under \emph{alpha-aligned variability}, we drew independent variables $Z_{sj}\in\{-1,1\}$ for subject $s$ and node $j$, and set
\[
R^{(s)}_{jk}=R_{jk}(1+\delta Z_{sj})(1+\delta Z_{sk}), \qquad j\ne k,
\]
with diagonal elements set to one. This mechanism introduces subject-specific variation directly in the multiplicative structure. We used $\delta\in\{0,0.08,0.16\}$.

Under \emph{generic correlation variability}, each subject received an independent symmetric, zero-diagonal random perturbation $B_s$. The perturbation was scaled so that the eigenvalues of $R^{-1/2}B_sR^{-1/2}$ lay in $[-1,1]$, ensuring that
\[
R^{(s)}=R+\delta B_s
\]
remained positive definite for all simulated variability levels. Because the distribution of $B_s$ is sign-symmetric, these perturbations have mean zero and do not systematically alter the group-level correlation matrix. We used $\delta\in\{0,0.35,0.70\}$ for this mechanism.

Conditional on each subject's population correlation matrix $R^{(s)}$, we generated a sample covariance matrix from a Wishart distribution with $T=120$ degrees of freedom and converted it to a sample correlation matrix. We fit the multiplicative model to each simulated dataset and computed the diagonal entries of the Schur-complement variance estimator for $\hat{\alpha}$.

For each condition, we repeated the simulation $M=100$ times. For every coordinate $\alpha_j$, we calculated the empirical variance of $\hat{\alpha}_j$ across repetitions and the mean estimated GEE variance. Each point displayed in Figure \ref{fig:sub_variability} compares the averages of empirical and estimated standard deviations over the $p$ coordinates and 3 designs.

\section{Supplementary Tables and Figures}

\begin{figure}[!tbh]
\begin{center}
    \includegraphics[scale=0.7]{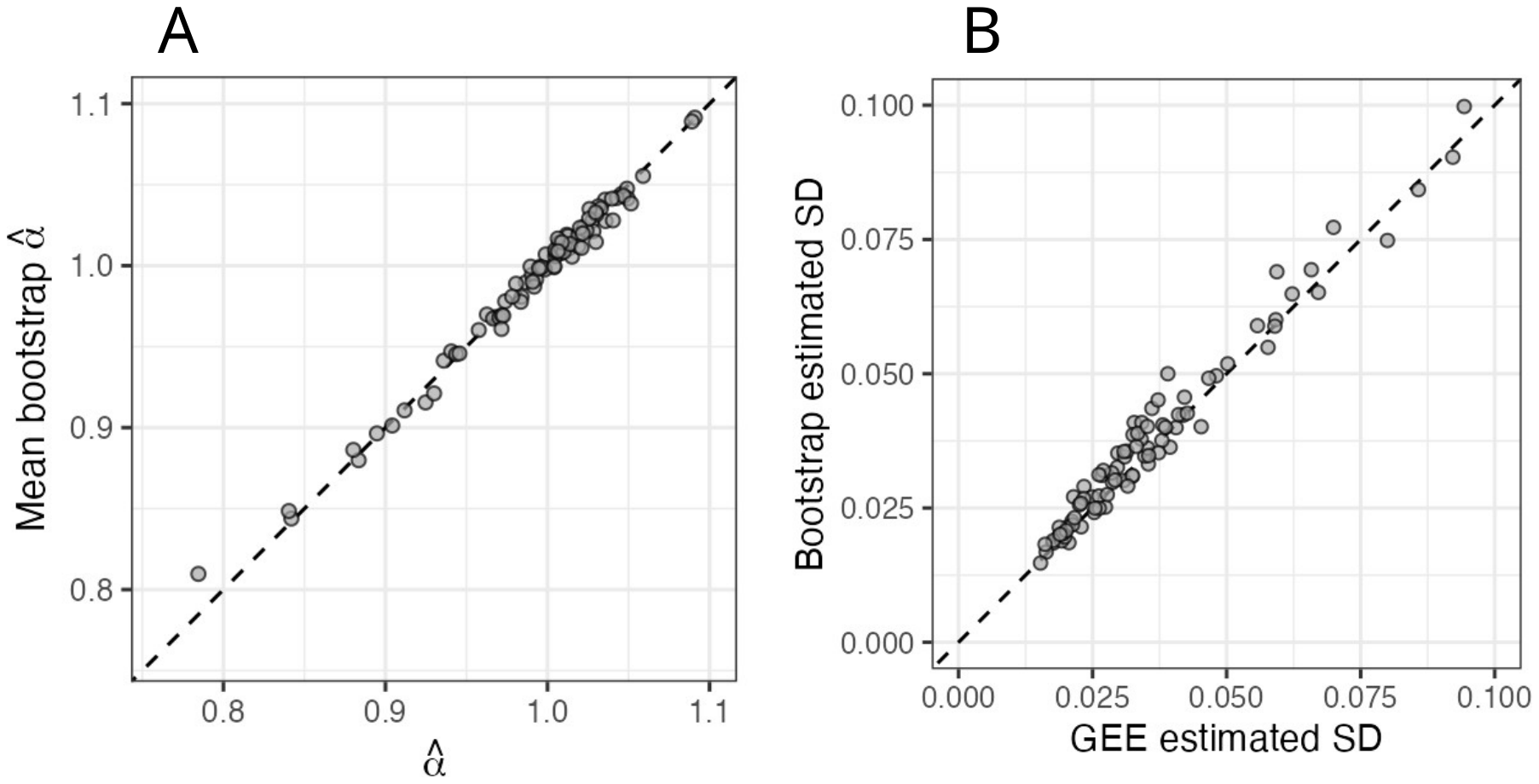}
    \caption{\textbf{Comparison of region TGA estimates for the original data with bootstrap estimates.}A.  TGA estimates of $\alpha_j$ (x-axis) on the original data compared with average of estimates derived from bootstrap replicates of the data. B. Estimates of the GEE standard deviation (x-axis) against empirical standard deviation of the bootstrap estimates. In both cases we see good agreement, though some discrepancy between the two estimates of the standard deviation is observed. Bootstrap replicates were obtained by sampling with replacement from each group separately.} 
\label{figure:bootstrap}

\end{center}
\end{figure}

\begin{figure}[!tbh]
\begin{center}
    \includegraphics[scale=0.12]{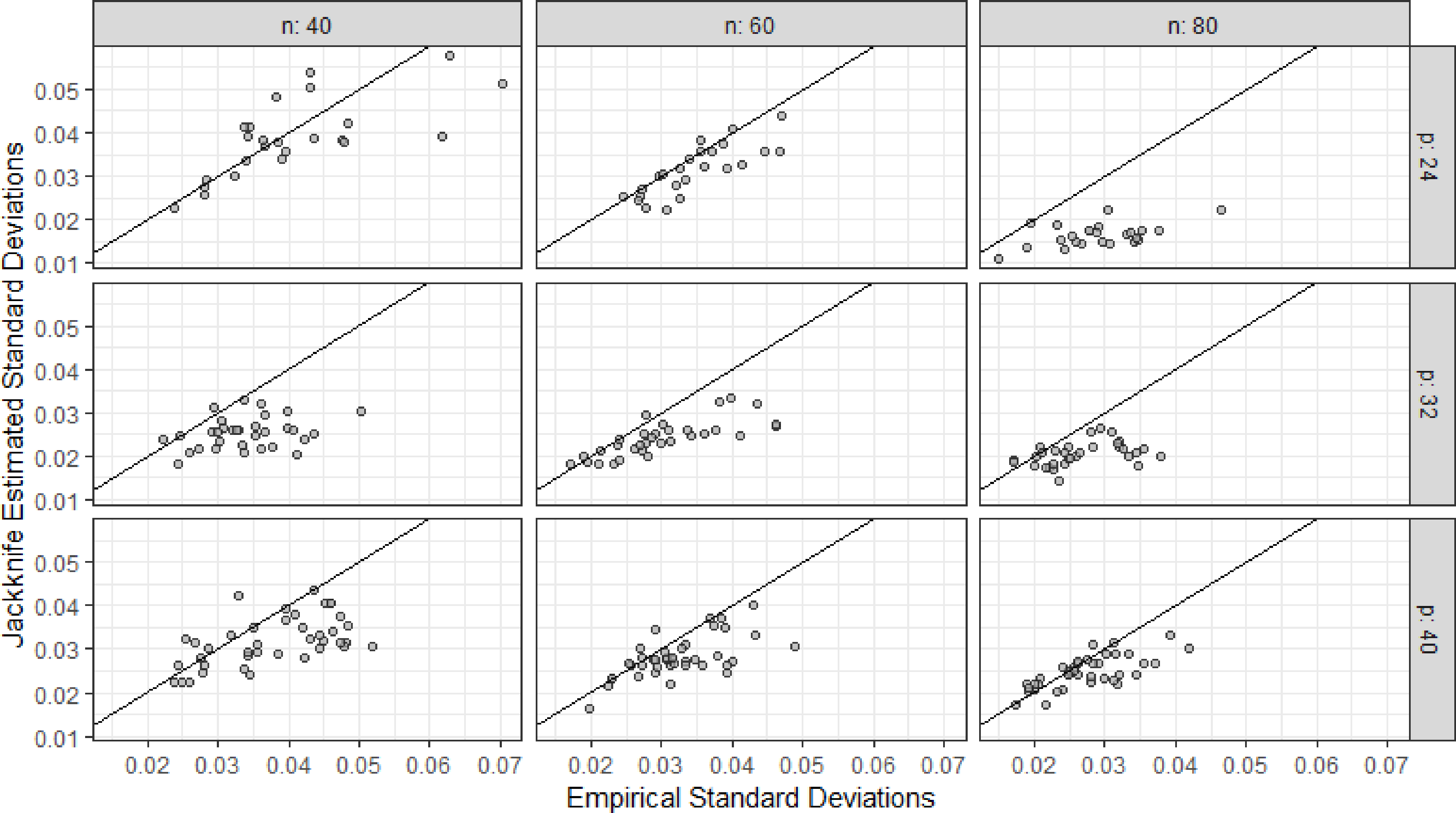}
    \caption{ \small
    \textbf{SD estimation with Jackknife.} Estimated standard deviation versus true values for various $n, p$ values.
    Simulations were conducted under the weak null and with positive auto-correlation.}
    \label{fig:jack-var}
\end{center}
\end{figure}

\begin{figure}[!tbh]
\begin{center}
    \includegraphics[scale=0.7]{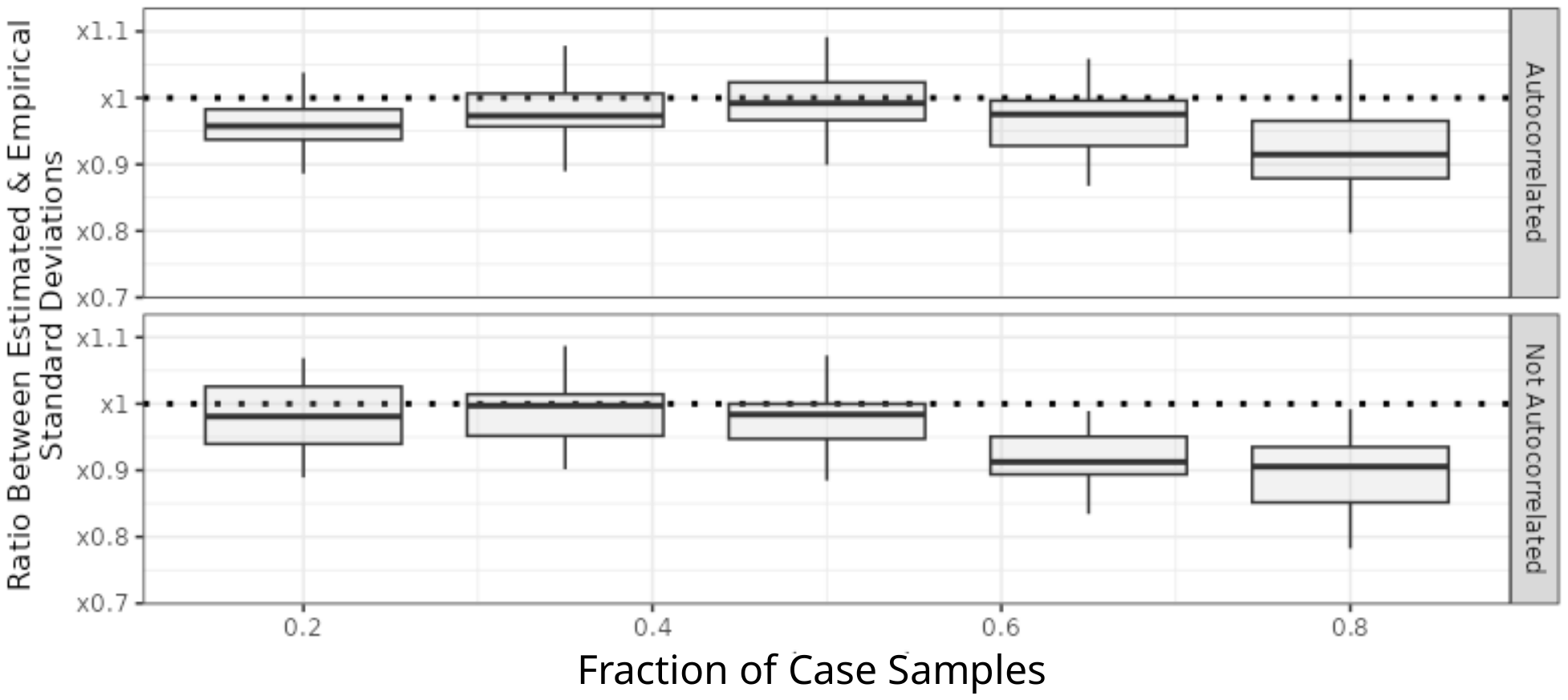}
    \caption{\small Effect of the relative size of the patient group ($\frac{n_{B}}{n_{A}+n_{B}}$) on the SD estimator (GEE). The bias is most pronounced when group B is large relative to group A. In all runs, $n=80$ and $p=35$ }
    \label{fig:percent-sick}
    \end{center}
    \end{figure}

\end{document}